\documentclass[journal, a4paper]{IEEEtran}
\usepackage{amsmath,amsfonts,amsthm}
\usepackage{algorithmic}
\usepackage{algorithm}
\usepackage{array}
\usepackage[caption=false,font=normalsize,labelfont=sf,textfont=sf]{subfig}
\usepackage{textcomp}
\usepackage{stfloats}
\usepackage{dsfont}
\usepackage{url}
\usepackage{verbatim}
\usepackage{graphicx}
\usepackage{mathtools}
\usepackage{cite}
\usepackage{hyperref}
\usepackage{siunitx}
\usepackage{tikz}
\usetikzlibrary{arrows}
\usetikzlibrary{shapes.multipart}
\usetikzlibrary{plotmarks}
\usetikzlibrary{patterns}

\graphicspath{{pics/}{plots/}}
\newtheorem{theorem}{Theorem}

\IEEEoverridecommandlockouts
\newcommand\copyrighttext{%
\footnotesize \textcopyright \enspace 2023 IEEE. Personal use of this material is permitted. Permission from IEEE must be obtained for all other uses, in any current or future media, including reprinting/republishing this material for advertising or promotional purposes, creating new collective works, for resale or redistribution to servers or lists, or reuse of any copyrighted component of this work in other works. DOI: \href{https://doi.org/10.1109/JIOT.2023.3245727}{10.1109/JIOT.2023.3245727}
}
\newcommand\copyrightnotice{%
\begin{tikzpicture}[remember picture,overlay]
\node[anchor=south] at (current page.south) {\fbox{\parbox{\dimexpr\textwidth-\fboxsep-\fboxrule\relax}{\copyrighttext}}};
\end{tikzpicture}%
}

\begin{document}

\title{On the Limits and Best Practice for NB-Fi:\\a New LPWAN Technology}

\author{\IEEEauthorblockN{
	Dmitry Bankov\IEEEauthorrefmark{1},~\IEEEmembership{Member,~IEEE},
	Polina Levchenko\IEEEauthorrefmark{1},\\%\IEEEauthorrefmark{2},\\
	Andrey Lyakhov\IEEEauthorrefmark{1},~\IEEEmembership{Member,~IEEE},
	Evgeny Khorov\IEEEauthorrefmark{1},~\IEEEmembership{Senior Member,~IEEE}
	% <-this % stops a space
}

\IEEEauthorblockA{\IEEEauthorrefmark{1}Institute for Information Transmission Problems of the Russian Academy of Sciences, Moscow, Russia\\
	%\IEEEauthorrefmark{2}Moscow Institute of Physics and Technology, Moscow, Russia\\
	Email: \{bankov, levchenko, lyakhov, khorov\}@wireless.iitp.ru}
}

\maketitle
\copyrightnotice

\begin{abstract}
NB-Fi is a new low-power wide-area network technology, which has become widely used for Smart Cities, Smart Grids, the Industrial Internet of Things, and telemetry applications.
Although many countries use NB-Fi, almost no papers study NB-Fi, and its peak performance is unknown.
This paper aims to fill this gap by analyzing this technology and studying the problem of rate assignment in NB-Fi networks.
For that, the paper develops a mathematical model used to find the packet loss ratio, packet error rate, and the average delay for various rate assignment approaches.
The performance evaluation results are used to develop the guidelines for NB-Fi configuration to optimize the network performance.
\end{abstract}


\section{Introduction}
\label{sec:intro}
Low Power Wide Area Networks (LPWANs) are used in many Wireless Internet of Things (IoT) systems. Although they provide much lower throughput than cellular technologies or Wi-Fi, LPWAN technologies are simple and easy to deploy. At the same time, LPWAN technologies provide greater coverage than Wi-Fi or RFID \cite{huang2020freescatter}. Many LPWAN technologies were designed by small and medium-sized enterprises (SMEs) as national-level ones and then expanded their market worldwide. For example, Sigfox~\cite{goursaud2015dedicated} and LoRa~\cite{vangelista2015long} were initially designed in France but later became used all over the world~\cite{sigfox_coverage, lorawan_coverage}.

The NB-Fi protocol~\cite{nbfi_standard} developed by WAVIoT \cite{waviot} is the first Russian LPWAN ultra narrowband (UNB) technology with an open standard. NB-Fi networks have been widely deployed in Russia~\cite{waviot_russia}, France~\cite{waviot_france}, Serbia~\cite{waviot_serbia}, India~\cite{waviot_india}, Argentina~\cite{waviot_argentina}, Moldova~\cite{waviot_moldova}, and Kazakhstan~\cite{waviot_kazahstan}.

As NB-Fi is a relatively new technology, only a few surveys briefly mention it and its basic operation parameters~\cite{smart_cities2019, petrenko2018iiot, ikpehai2018low}.
However, these surveys only provide the nominal parameters of NB-Fi, such as the maximal PHY layer bitrate, bandwidth, and transmission power, but they do not contain many important details on the protocol operation related to the Medium Access Control (MAC) and Transport layer, such as the description of frequency selection or packet retry algorithms and do not study the operation of NB-Fi networks in scenarios with multiple sensors interfering with each other.
Because of the lack of detailed information about the NB-Fi operation, its performance in typical IoT scenarios is unknown to the academic community.
At the same time, the performance of NB-Fi cannot be estimated with the models or simulations of other well-studied LPWAN technologies.
Although  NB-Fi has much in common with such popular LPWAN technologies as Sigfox~\cite{mekki2018overview} (they both use the UNB modulation) or LoRaWAN~\cite{bankov2016limits, bankov2019lorawan} (they have similar operation modes), the peculiarities of NB-Fi make it impossible to directly apply the results obtained in studies of other LPWAN technologies to NB-Fi networks.
Thus, new mathematical models of NB-Fi are needed to evaluate its performance.

NB-Fi devices can transmit their data at four possible bitrates, assigned individually to each device. Although the bitrates have \emph{the same} spectral efficiency (bit/Hz/s), the transmissions have different reliability, duration, and bandwidth. So a significant problem is how to assign the bitrates in order to reduce the packet loss rate (PLR) and/or the average delay. On the one hand, the devices could use the fastest bitrate, but on the other hand, such transmissions occupy wider channels and thus may increase collision probability. Therefore, a more detailed study is required.

The novelty and contribution of this paper are threefold.
First, we provide the first comprehensive introduction and analysis of NB-Fi, focusing on its original features related to channel access, such as the frequency selection, packet retry and acknowledgment algorithms, and many other MAC and Transport layer features.
Second, we develop a mathematical model of an NB-Fi network to investigate its limits, i.e., to find how the packet error rate (PER), packet loss ratio (PLR), and the average delay depend on the load and the network parameters.
The developed model is a model of an asynchronous, both in time and frequency domain, ALOHA~\cite{aloha} with packets that can occupy various amounts of time and frequency resources depending on their bitrates, and with a specific retry policy used in NB-Fi. 
Third, we provide guidelines on NB-Fi bitrate allocation to minimize the PLR or the average delay.
Optimizing PLR is important because it improves the network capacity, which is an essential metric for LPWANs that determines how many devices can operate in a network, while minimizing the average delay decreases power consumption, which is crucial for battery-powered sensors.

The rest of the paper is organized as follows.
Section~\ref{sec:nb-fi} introduces NB-Fi.
In Section~\ref{sec:related_works}, we review related papers.
Section~\ref{sec:problem} describes the scenario and the problem statement.
In Section~\ref{sec:model}, we describe the developed mathematical model of the NB-Fi network.
Section~\ref{sec:results} presents and discusses numerical results.
The conclusion is given in Section~\ref{sec:conclusion}.
Appendices~\ref{sec:frame_format},~\ref{sec:transport_format}, and~\ref{sec:rate_control} provide more information about NB-Fi.

\section{An Overview of NB-Fi}
\label{sec:nb-fi}

NB-Fi standard defines operation at the physical, data link, and transport layer. In other words, it describes modulation and coding schemes, channel access, frame formats, data transmission sequences, etc.

NB-Fi uses narrow channels in unlicensed ISM radio bands that limit the emitted power.
For example, in Russia, NB-Fi operates at 868.7-869.2 MHz with the power limit of \SI{100}{\mW}.

A typical NB-Fi network has a ``star of stars'' topology and consists of a server, base stations (BSs), and end devices (further referred to as \emph{sensors}).
BSs are connected to the server via a broadband link.
Sensors communicate with the server via BSs using a wireless NB-Fi link.
In contrast to cellular technologies, sensors are not associated with a single BS, so all BSs that receive a frame from a sensor redirect the frame to the server. In the reverse direction, the server chooses which BS transmits a frame to the sensor.

\subsection{Modulation}
NB-Fi sensors communicate with BSs using differential binary phase-shift keying (DBPSK) in the uplink (UL) and DBPSK or binary phase-shift keying (BPSK) in the downlink (DL).
Although DBPSK has a higher bit error rate than BPSK \cite{proakis2001digital}, it does not require estimation of carrier phase and can be used even if the transmitter's oscillators of sensors are unstable.

Despite different nominal bitrates of 50, 400, 3200, or 25600 bps, all bitrates' spectral efficiency (bit/s/Hz) is the same.
As Fig.~\ref{fig:widths} and Table~\ref{tab:mcs_values} show, with a lower bitrate, the transmission becomes longer but occupies a narrower bandwidth, which increases the power spectral density and, consequently, the transmission range.
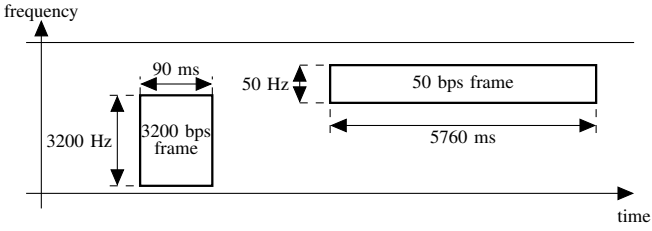
\begin{figure}[!tb]
\centering
\begin{tikzpicture}[scale=1]
	\draw [line width=0.3mm] (4.0, 1.4) rectangle (7.5, 1.9);
	\draw [line width=0.3mm] (1.5, 0.3) rectangle (2.45, 1.5);
	\node at (5.75,  1.65) {\scriptsize{50 bps frame}};
	\node at (1.975, 1.0) {\scriptsize{3200 bps}};
	\node at (1.975, 0.8) {\scriptsize{frame}};
	\draw (0, 2.2) -- (8, 2.2);
	\draw [dashed] (1.1, 0.3) -- (1.5, 0.3);
	\draw [dashed] (1.1, 1.5) -- (1.5, 1.5);
	\draw [arrows={triangle 45-triangle 45}] (1.2, 0.3) -- (1.2, 1.5);
	\node at (0.7,  0.9) {\scriptsize{3200 Hz}};
	\draw [dashed] (1.5, 1.5) -- (1.5, 1.75);
	\draw [dashed] (2.45, 1.5) -- (2.45, 1.75);
	\draw [arrows={triangle 45-triangle 45}] (1.5,1.65) -- (2.45,1.65);
	\node at (1.975,  1.85) {\scriptsize{90 ms}};
	
	\draw [dashed] (3.6, 1.4) -- (4.0, 1.4);
	\draw [dashed] (3.6, 1.9) -- (4.0, 1.9);
	\draw [arrows={triangle 45-triangle 45}] (3.6, 1.4) -- (3.6, 1.9);
	\node at (3.15,  1.65) {\scriptsize{50 Hz}};
	\draw [dashed] (4.0, 1.0) -- (4.0, 1.4);
	\draw [dashed] (7.5, 1.0) -- (7.5, 1.4);
	\draw [arrows={triangle 45-triangle 45}] (4.0, 1.1) -- (7.5, 1.1);
	\node at (5.75,  0.95) {\scriptsize{5760 ms}};
	\draw [arrows={-triangle 45}] (0, 0.2) -- (8, 0.2);
	\node at (8, -0.1) {\scriptsize{time}};
	\draw [arrows={-triangle 45}] (0.2,0) -- (0.2,2.5);
	\node at (0.2, 2.6) {\scriptsize{frequency}};
\end{tikzpicture}
\caption{Different frame widths and durations}
\label{fig:widths}
\end{figure}
For each bitrate, the standard \cite{nbfi_standard} estimates the required receiver sensitivity $S_{crit}$, shown in Table \ref{tab:mcs_values}, as follows:
\begin{equation}
S_{crit} = 10 \log_{10} kT \Delta + N_{base} + SNR_{BER},
\label{eq: S_crit}
\end{equation}
where  $k$ is the Boltzmann constant, $T = \SI{290}{\K}$ is the temperature, $\Delta$ is the estimated frame bandwidth, $N_{base} = \SI{2}{\dB}$ is the input noise, $SNR_{BER}= \SI{5}{\dB}$ is the signal-to-noise ratio (SNR) required to achieve the required bit error rate ($BER = 10^{-5}$).

\begin{table*}[!t]
\centering
\caption{Constants Depending on Bitrate}
\begin{tabular}{|c|c|c|c|c|c|c|c|c|c|}
	\hline
	BN & Bitrate, & Frequency band & Frame duration, & Sensitivity, & $T_{delay}$, ms & $T_{listen}$, ms & $T_{rnd}$, ms & $SNR_{RX/TX}$,  & Max distance  \\ 
	&  bps             & $\Delta$, Hz & ms              & dBm          & & & & dB & $R^*_{ i}$, km\\ \hline
	1  & 50    & 50    & 5760 & -150 & 5900 & 60000 & 5000 & 0  & 10.869\\
	2  & 400   & 400   & 720  & -141 & 740  & 30000 & 1000 & 9  &  6.023\\
	3  & 3200  & 3200  & 90   & -132 & 95   & 6000  & 100  & 18 &  3.337\\
	4  & 25600 & 25600 & 11.25   & -123 & 15   & 6000  & 100  & 27 &  1.849\\
	\hline
\end{tabular}
\label{tab:mcs_values}
\end{table*}

\subsection{Operation Modes}
\label{sec:modes}
NB-Fi sensors can operate in three modes.

In the continuous RX mode (CRX), the sensors always listen to the channel and can transmit and receive data at any time, which minimizes delays but consumes too much energy.

The No RX mode supports only UL communications but reduces energy consumption.
The sensor sleeps, i.e., its radio is always switched off, except for the intervals when it transmits data.

In the discontinuous RX mode (DRX), after the end of each UL transmission, the radio remains on for an interval $T_{listen}$ during which it listens for DL frames, see Fig.~\ref{fig:retry_timings}.
The server buffers all data destined for DRX sensors and transmits them during this time interval when the appropriate sensor is listening to the channel.

By default, NB-Fi devices access the channel using an ALOHA-like approach: when a device has a frame, the device transmits it without listening to the channel. At the same time, the device shall comply with the regulatory duty cycle restrictions (the corresponding rules are out of the scope of the standard) and the retransmission rules (described in Section \ref{sec:acks}).

Also, the devices may implement the Listen Before Talk (LBT) policy: they listen to the channel before transmission and postpone transmission until the channel becomes idle instead of limiting the duty cycle.

\subsection{Central Frequency Selection}
\label{sec:frequency_selection}
%%%%%%%%%%%%%%

The NB-Fi operator needs to define UL and DL frequency bands that do not intersect. Specifically, it selects the base (central) frequency $F_{base}$ and operating bandwidth, the minimum value for which is \SI{51.2}{\kHz} for UL, and  \SI{102.4}{\kHz} for DL.
Then the operator splits each band into several subbands, which are assigned to various groups of stations.
Each subband is defined by two integer parameters: $0 \le W_{UL} \le 7$ and $-63 \le O_{UL} \le 63$, which specify the width of the subband $B_{UL} = 6400\times2^{W_{UL}}$ and the offset of subband central frequency $O_{UL, band} = B_{UL} \times O_{UL}$ with respect to $F_{base}$. 

Within the allocated subband, the sensor calculates the central frequency $f$ for a frame as a function of the bitrate, the frame parity ($p$), which equals 0 or 1 and changes the value for each frame, sensor identifier ($id$), and the three least significant bytes of the packet's message authentication code ($MIC0\_7$, see Appendix~\ref{sec:frame_format}):

\begin{equation}
f = F_{base} +  O_{UL, band} + O_{UL, f}
\label{eq:f_uplink}
\end{equation}
where $O_{UL, f}$ determines a pseudo-random position of the central frequency $f$ within the subband: 
\begin{equation}
\label{eq:uplink_channel}
O_{UL, f} = (-1)^{p + 1}G_{UL}\frac{(id + MIC0\_7)\text{ mod } 256 }{255},
\end{equation}
and $G_{UL}$ determines the ``effective'' space for selecting such a position:
\begin{equation}
G_{UL} = \begin{cases}
	\frac{B_{UL} - 2\times \Delta - 2000}{2}, & B_{UL} > 2\Delta + 2000, \\
	0, & B_{UL} \le 2\Delta + 2000.
\end{cases}
\end{equation}
In other words, if the frame is too wide with respect to the subband, $G_{UL} = 0$, and the frame is transmitted at the central frequency of the subband. Otherwise, the frame is located pseudo-randomly within the subband taking into account some guard band.

The central frequency of transmission in the DL is calculated similarly to the UL one with two modifications: $O_{DL, channel}$ does not depend on $MIC0\_7$ and its sign depends on the sensor identifier $id$ instead of the frame parity $p$. Therefore, all frames addressed to a sensor are transmitted at a fixed central frequency, simplifying the sensors' receiving chain.

\subsection{Acknowledgments and Retransmissions}
\label{sec:acks}

In NB-Fi, the acknowledgment procedure belongs to the transport layer.
It is rather flexible and can be configured independently for each sensor or even switched off. 
To reduce overhead, the procedure allows the receiver to send an acknowledgment frame not after every packet, but after a batch of packets of size $2^n$ ($n = 0, ..., 5$), where $n$ is a configurable parameter. To identify packets within the batch, the devices use the ITER field, see Fig.~\ref{fig:header}.

Once a device receives a frame with the flag ACK, it has to send back a service frame $ACK\_P$ that contains a mask indicating the received frames. 
Having received $ACK\_P$, the sender retransmits the lost frames using the same ITER values as in the previous transmission attempt. 

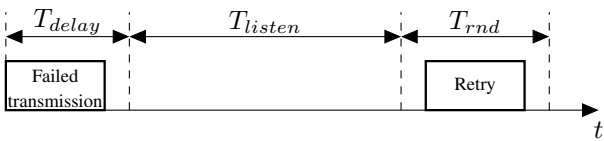
\begin{figure}[!b]
\vspace{-1em}
\centering
\begin{tikzpicture}[scale=0.65]
	\draw [arrows={-triangle 45}] (0,1) -- (12,1);
	\draw [arrows={triangle 45-triangle 45}] (0,2.5) -- (2.5,2.5);
	\draw [arrows={triangle 45-triangle 45}] (2.5,2.5) -- (8,2.5);
	\draw [arrows={triangle 45-triangle 45}] (8, 2.5) -- (11,2.5);
	\draw [line width=0.3mm] (0, 1) rectangle (2.0, 2);
	\draw [dashed] (0, 1) -- (0, 3); 
	\draw [dashed] (2.5, 1) -- (2.5, 3);
	\draw [dashed] (8, 1) -- (8, 3);
	\draw [dashed] (11, 1) -- (11, 3);
	\draw [line width=0.3mm] (8.5, 1) rectangle (10.5, 2);
	\node at (1.00,  1.7) {\scriptsize{Failed}};
	\node at (1.00,  1.2) {\scriptsize{transmission}};
	\node at (9.50,  1.5) {\scriptsize{Retry}};
	\node at (12.0,  0.6) {$t$};
	\node at (1.25,  2.8) {$T_{delay}$};
	\node at (5.25,  2.8) {$T_{listen}$};
	\node at (9.5,   2.8) {$T_{rnd}$};
\end{tikzpicture}
\vspace{-1em}
\caption{Retransmission Timings}
\label{fig:retry_timings}
\end{figure}

If the sender does not receive $ACK\_P$, it makes a retry after a random backoff time uniformly distributed within the window $\left[T_{delay} + T_{listen}, T_{delay} + T_{listen} + T_{rnd}\right]$ (see Fig.~\ref{fig:retry_timings}), where $T_{delay}$, $T_{listen}$, and $T_{rnd}$ depend on the UL bitrate (see Table~\ref{tab:mcs_values}).
Note, that this window is defined in such a way that if two frames collide and one of them is transmitted at 50 or 400 bps, and the other one is transmitted with a different bitrate, the repetitive collision of the same frames is impossible. 
Also note, that since $MIC0\_7$ depends on the contents of the frame and Crypto Iter (see Appendix~\ref{sec:frame_format}) is iterated each retry, the $MIC0\_7$ changes, too. Consequently, the sender selects a new central frequency for each transmission attempt, see \eqref{eq:uplink_channel}.

\section{Related Works}
\label{sec:related_works}
NB-Fi is a rather new technology and has not been carefully studied in the literature yet.

The paper~\cite{smart_cities2019} briefly compares LoRaWAN, Sigfox, NB-Fi, Nwave, and RPMA, considering their white-sheet properties, such as the maximal bitrate, transmission range, number of devices connected to a BS, etc. Thus, that paper only shows the nominal characteristics written in specifications, while the real network performance depends on many factors and requires a more detailed study in different scenarios.

A similar study is present in~\cite{ikpehai2018low}, which compares NB-Fi, LoRaWAN, Sigfox, RPMA, LTE-M, and NB-IoT.
This study provides nominal parameters of NB-Fi related to its PHY layer, such as the antenna gains, transmission power, and link budget, and uses the Okumura-Hata model \cite{hata1980empirical} to calculate the network coverage and energy consumption of devices with such parameters.
However, the paper contains only a PHY-layer study and does not provide any details on NB-Fi protocol, its channel access features, and how the NB-Fi network would perform in case of numerous sensors deployment.

The paper~\cite{levchenko2022performance} compares NB-Fi, Sigfox, and LoRaWAN by PLR, PER, and the average delay in different scenarios. The authors analyze network performance using simulation in specific scenarios in order to provide guidelines to decide which technology to use under which conditions. Using the developed simulation the authors come to the conclusion that NB-Fi achieves lower PLR than Sigfox and LoRaWAN in a scenario where sensors transmit small pieces of data in the most reliable way. However, the authors do not propose any optimizations to increase network performance indicators.

In the paper~\cite{petrenko2018iiot}, the authors describe a general IoT infrastructure and industrial IoT (IIoT) as a significant part of it.
The authors discuss that NB-Fi can be used in IIoT as an alternative to NB-IoT, LTE-M, LoRa, and Sigfox. However, the efficiency of NB-Fi is not studied.

Some parts of the NB-Fi technology are similar to other LPWANs. For example, both NB-Fi and Sigfox use a version of asynchronous time-frequency ALOHA channel access, taking into account the instability of sensors' oscillators. A mathematical model of Sigfox channel access is developed in \cite{li20172d}, where the authors use stochastic geometry methods to describe the time-frequency interference and show how to find the outage probability and throughput with this model.
However, the developed model describes only the case when all frames have the same duration and channel width, so this model cannot be used for NB-Fi, where devices can use different bitrates and the frames have variable duration and width.
Also, the devices generate saturated traffic in the considered model, which is not a typical scenario for sensor networks.

A recent paper on NB-Fi~\cite{bankov2022performance} analyzes the channel access method in NB-Fi using simulation and points out the factors that degrade the performance of NB-Fi networks. However, the authors do not propose a solution for the found problems and do not provide any recommendations for increasing the network's performance. In contrast, in our paper, based on the developed mathematical model, we obtain many results that guide us to provide recommendations on improving the performance of NB-Fi. 

The paper~\cite{pavlova2022efficiency} studies the efficiency of carrier sense multiple access in NB-Fi networks and its impact on energy consumption. The authors consider the Listen Before Talk (LBT) requirement for operating in unlicensed ISM radio bands imposed in some countries. However, in NB-Fi, the LBT mode (described in Section~\ref{sec:modes}) is optional, is more difficult to implement, requires more battery power, and, thus, to the best of our knowledge is not supported by many devices. Therefore, most of the devices operate using an ALOHA-like approach complying with the duty cycle restrictions, which are considered in this paper.

The NB-Fi channel access is also similar to the LoRaWAN one because of a palette of bitrates and the retry procedure. Mathematical models of LoRaWAN channel access are developed in \cite{bankov2017pimrc, bankov2019lorawan,  capuzzo2018mathematical}. In \cite{bankov2017pimrc}, the authors reveal that in LoRaWAN, the retries cannot be described by a Poisson stream even if data frames are generated according to a Poisson process.

Also, the developed model takes into account that the transmissions at different bitrates are orthogonal. The latter property is not valid for NB-Fi.
In \cite{bankov2019lorawan}, the authors expand the model from \cite{bankov2017pimrc} to a case when data frames can be lost not only due to collisions but also due to random noise in the channel. They also study how to assign bitrates to sensors in order to satisfy the heterogeneous quality of service requirements in a LoRaWAN network. 

In \cite{capuzzo2018mathematical}, the authors extend \cite{bankov2017pimrc} to consider both UL and DL traffic. They also consider that the transmissions at different bitrates can collide if the difference in the receive power is sufficiently high.

Despite the similarity, NB-Fi has some features that do not allow using results obtained for LoRaWAN. 
First, while LoRaWAN frequency channels can be considered orthogonal~\cite{huang2020loradar}, in NB-Fi, the transmissions can randomly intersect in the frequency domain. 
Second, in LoRaWAN, acknowledgment, and retry timings are the same for all bitrates, while in NB-Fi, they are different for various bitrates.

The problem of bitrate assignment for LoRaWAN has been considered in numerous papers \cite{bankov2019lorawan, bankov2020algorithm, cuomo2017explora, abdelfadeel2018fair, croce2017impact, zorbas2018improving}, which propose centralized and distributed approaches to assign bitrates to devices in order to improve the network performance.
The same problem is actual for NB-Fi networks as well, but the solutions developed for LoRaWAN cannot be used for NB-Fi because of the difference in modulations.
As mentioned above, in LoRaWAN, signals with different bitrates can be considered orthogonal, while in NB-Fi, they do interfere.
Second, in LoRaWAN, the usage of different bitrates does not affect the bandwidth occupied by the signal, while in NB-Fi, the change of bitrate changes the bandwidth and thus affects the collision probability.

In this paper, we consider the problem of bitrate assignment for NB-Fi devices, taking into account its discrepancies with the other LPWAN technologies. We design a new mathematical model of NB-Fi channel access that allows us to find the PLR and the average delay as functions of the NB-Fi network and traffic parameters. Then we use the model to evaluate various strategies to assign the bitrates to devices in order to optimize the delay and PLR.

\section{Scenario and Problem Statement}
\label{sec:problem}
Consider a network with a server, $N = 1000$ sensors distributed uniformly within a circle with a radius $R_1$, and a BS placed in the center of this circle.
The sensors and BS are static.
The sensors operate in the DRX mode and transmit frames to the server via the BS.
The sensors generate new frames according to the Poisson stream with a total intensity $\lambda$.
There is no traffic in the DL except for acknowledgments sent by the BS for each successful frame from sensors.
We assume that the sensors are simple devices with limited memory and thus can buffer only one frame at a time.
If a sensor generates a new frame while it transmits another frame, then the new frame preempts the existing one, if it fails. 
The maximal number of transmission attempts per frame is $RL = 7$.
All frames are 36-byte long as specified in the NB-Fi standard.

All devices use the transmission power of $E_T$.
The frames are received successfully if the signal-to-interference-and-noise ratio (SINR) in all parts of the frame exceeds some threshold $\nu$ for the entire duration of the frame.

In the paper, we study how to allocate bitrates to the sensors to minimize the average delay and PLR in the NB-Fi network. 
At first sight, the sensors could use the standardized bitrate allocation algorithm described in Appendix \ref{sec:rate_control}, which tries to maximize the bitrate for all the devices. 
However, the usage of this algorithm may be inefficient as the fastest bitrates occupy more bandwidth and thus may increase collision probability even for devices with the other bitrates. 
To study this problem, we consider that the bitrates are assigned to the sensors during network initialization in some manner and do not change with time. Following the idea of the standardized bitrate allocation algorithm, 
we define signal power thresholds $S_{i}$, $i = 1, 2, 3, 4, 5 $ such that $S_{5} = \infty$ and $ \forall i<5, S_{i} \ge S_{i - 1}$.
Let the bitrate number (BN) $i$ be assigned to a sensor if its average signal power $S$ satisfies the inequality $S_{i} \leq S < S_{i+1}$.
We assume that the network is static, so the average sensor's signal power at the BS is determined by its distance to the BS.
Thus, by controlling the power thresholds $S_{i}$, we allocate bitrates to appropriate portions $p_i$ of sensors.
The acknowledgments are transmitted using the same bitrates as the corresponding frames in the UL.

Within such a scenario, we state the problems \emph{to find the bitrate distribution $p_i$ among the sensors such that minimizes PLR} and \emph{to find the bitrate distribution $p_i$ among the sensors such that minimizes the delay}.
Notably, delay minimization is an important issue for sensor networks not only for the sake of delay itself but also because we reduce the time the sensor is on and transmits its data or listens to the channel waiting for the acknowledgment. Consequently, we reduce energy consumption.

\section{Mathematical Model}
\label{sec:model}
Here and further, we consider that the sensors operate in the same subband, i.e., they have the same $F_{base}$ and $O_{UL, band}$ (see \eqref{eq:f_uplink}).
Let us pick a sensor that transmits its packet at BN $i$.
Let us find the probability that the packet is delivered successfully.
According to the preliminary simulation results, the success probability of the initial transmission attempt significantly differs from that of a retry, while the success probability of the initial and consecutive retries is almost the same. 
Therefore, we derive the probabilities for the initial transmission attempts and for retries separately. 
Also, we consider that in real-life scenarios of interest, the collision rate should be low.
It means that the traffic is not very heavy, and the number of retries is much less than the number of initial transmission attempts.

We describe the mathematical model as follows.
In Section \ref{sec:initial}, we derive the probability of success during an initial transmission attempt.
To calculate it, we need the probability of the sensors' signal to have an SINR below $\nu$ which, in its turn, depends on the distribution of frequency difference of simultaneously transmitted packets.
We find these values in Sections \ref{sec:prob_freq} and \ref{sec:interference}.
After that, in Section \ref{sec:retry}, we derive the probability of successful transmission when a sensor makes a retry.
With the successful transmission probabilities, we find PLR in Section \ref{sec:plr} and the average delay in Section \ref{sec:delay}.
Finally, in Section \ref{sec:optimization}, we show how to use the developed model to optimize the PLR or delay.

\subsection{Initial Transmission Attempt}
\label{sec:initial}
In this Section, we derive the probability of success during an initial attempt.
We split it into the probability of successful transmission of a packet and an acknowledgment, the first one is found using the properties of Poisson streams, while the second one is estimated under an assumption of low traffic intensity.

A packet is delivered successfully, and the sensor stops sending it once it is received by the BS, and the corresponding acknowledgment is received by the sensor.
Thus, the success probability for the initial transmission attempt with the BN $i$ equals 
\begin{equation}
\label{eq:psfirst}
P^{S}_{ini,i} = P^{Data}_{ini, i} P^{Ack}_{ini, i},
\end{equation}
where $P^{Data}_{ini, i}$ is the success probability for a packet initial transmission attempt at BN $i$, and $P^{Ack}_{ini, i}$ is the success probability for the corresponding acknowledgment, provided that the data is successfully delivered.

Consider a single BS scenario, so no packets collide in the DL channel.
We assume that with the selected bitrate, the transmission is reliable, so an acknowledgment can be lost only if the BS cannot transmit it within the $T_{listen}$ interval after the reception of the packet because the BS has many pending acknowledgments in the buffer which have to be transmitted on the same frequency.
However, the value of $T_{listen}$ (see Table \ref{tab:mcs_values}) is much higher than the packet duration. So taking into account the low traffic rate, the possibility that the BS cannot deliver ACK within $T_{listen}$ is negligible and $P^{Ack}_{ini, i} = 1$.

Thus, the average PER for the initial transmission attempt equals:
\begin{equation}
\label{eq:perfirst}
PER_{ini} = 1-\sum_{i = 1}^{4} p_i P^{S}_{ini, i}=1-\sum_{i = 1}^{4} p_i P^{Data}_{ini, i}.
\end{equation}

Let us find $P^{Data}_{ini, i}$.
The packet delivery ratio shall be high enough in typical IoT scenarios of interest. Consequently, as the number of retries is much smaller than the number of the initial transmission attempts, we neglect the influence of retries on the initial transmission of the packet as well as packet drops because of non-empty buffers.
Given the flow intensity of all packets generated by all the sensors $\lambda$, the intensity of the packets at BN $i$ equals $\lambda_i = \lambda p_i$.
A transmission attempt of a packet is unsuccessful if at least one packet intersects the considered one in time and frequency, and the induced interference power is high enough.
$P^{Data}_{ini, i}$ is the probability of an opposite event:
\begin{equation}
\label{eq:pdata}
P^{Data}_{ini, i} = \smashoperator{\sum_{\substack{0 \leq k_j \leq N_j \\ j=1,...,4}}} Q^{freq}_{i, k_1, k_2, k_3, k_4} \prod_{j=1}^{4} \frac{\lambda^{k_j}_j (T_i + T_j)^{k_j}}{k_j!}e^{-\lambda_j \left(T_i + T_j\right)} ,
\end{equation}

Here, $N_j$ is the number of sensors using BN $j$. So, we sum over the possible numbers of \textit{other} sensors $k_j$ that use BN $j$ and generate a packet during the time interval $[-T_j, T_i]$, where $t = 0$ is the start of the considered packet transmission and $T_i$ is the duration of a packet at BN $i$.
Inside the sum, we multiply 
\begin{itemize}
\item $Q^{freq}_{i, k_1, k_2, k_3, k_4}$ which is the success probability of the considered packet provided that $k_1, ..., k_4$ packets are generated at corresponding BNs during the interval $[-T_j, T_i]$ and possibly interfere with it,\footnote{In other words, $Q^{freq}_{i, k_1, k_2, k_3, k_4}$ considers packet intersection in frequency and the events that the SINR is below $\nu$ (\SI{7}{\dB} in our scenario).} and
\item the product of probabilities that in a Poisson flow, $k_j$ packets are generated during the interval $[-T_j, T_i]$.
\end{itemize}

Obviously, $Q^{freq}_{i, 0, 0, 0, 0} = 1$, i.e., the packet is successful if no other packets interfere.

In the general case, the exact formula for $Q^{freq}_{i, k_1, k_2, k_3, k_4}$ is rather complex and requires consideration of many types of packet intersections.
To simplify the calculations, we assume that the interference from a packet on the received packet can be considered independently of the other interfered packets.
With such an assumption, we introduce $Q_{i, j}$ as the success probability of the considered packet transmitted at the BN $i$ provided that another packet at the BN $j$ is generated during the interval $[-T_j, T_i]$. It yields us the following theorem.

\begin{theorem}
\label{th:pdata}
If the packet intersections are independent of the other packets and the number of sensors is infinite ($N_j \to \infty, j=1,...,4$), then
\begin{equation}
	P^{Data}_{ini, i} = e^{-\sum_{j = 1}^{4}\lambda_j (T_{i} + T_j)(1 - Q_{i,j})}.
	\label{eq:theorem}
\end{equation}
\end{theorem}

\begin{proof}
If packet intersections can be considered independently of each other, then the success probability can be factorized:
\begin{equation}
	Q^{freq}_{i, k_1, k_2, k_3, k_4} = \prod_{j = 1}^{4} Q_{i, j}^{k_j},
\end{equation}
With such an assumption, we obtain
\begin{equation}
	\label{eq:pdatainii}
	P^{Data}_{ini, i} = \sum_{\substack{0 \leq k_j \leq N_j \\ j=1,...,4}}\prod_{j=1}^{4} \left(\frac{\left(\lambda_j (T_i + T_j) Q_{i, j}\right)^{k_j}}{k_j!}e^{-\lambda_j \left(T_i + T_j\right)}\right),
\end{equation}
which can be simplified using the distributive property
\begin{equation}
	\sum_{k_1 = 0}^{N_1} \sum_{k_2 = 0}^{N_2} \sum_{k_3 = 0}^{N_3} \sum_{k_4 = 0}^{N_4} \prod_{j = 1}^{4} \mathcal{G}_j(k_j) = \prod_{j = 1}^{4} \sum_{k_j = 0}^{N_j} \mathcal{G}_j(k_j),
\end{equation}
where $\mathcal{G}_j(k_j)$ equals the content in  parentheses in \eqref{eq:pdatainii},
and the Taylor series for the exponent function
\begin{multline}
	\lim_{N_j \to \infty}\sum_{k_j = 0}^{N_j} \left(\frac{\left(\lambda_j (T_i + T_j) Q_{i, j}\right)^{k_j}}{k_j!}e^{-\lambda_j \left(T_i + T_j\right)}\right) \\
	= e^{\lambda_j (T_i + T_j) Q_{i, j}} \times e^{-\lambda_j \left(T_i + T_j\right)}.
\end{multline}
As a result, we obtain \eqref{eq:theorem}.
\end{proof}
We further approximate the probability of success with this upper bound.
To calculate it, we need to find $Q_{i, j}$, for which we require the distribution of the difference of central frequencies of simultaneously transmitted packets.

\subsection{Frequency Difference Distribution}
\label{sec:prob_freq}
In this Section, we derive the distribution of the central frequency difference of two simultaneously transmitted packets.
For that, we solve a stochastic geometry problem: we consider all possible central frequencies of the packets and evaluate the resulting frequency differences and their corresponding probabilities.

Consider two packets transmitted during the intersecting time intervals at BNs $i$ and $j$. 
For simplicity, we denote them as packet $i$ and packet $j$, their central frequencies as $f_i$ and $f_j$, and the transmitting sensors as sensors $i$ and $j$, respectively.
Also, while the standard selects the central frequency for a packet within the interval $(F_{base} + O_{UL,band}-B_{UL} / 2, F_{base} + O_{UL,band}+ B_{UL} / 2)$, to simplify notation, we denote this interval as $(0, B_{UL})$.

As described in Section~\ref{sec:frequency_selection}, when a sensor chooses the frequency for transmission at BN $i$, it initially calculates a guard band $\omega_i = \Delta_i + \SI{1000}{\Hz}$, which depends on the packet bandwidth.
If two guard bands are wider than the allocated subband, i.e., the packets are ``wide'' with respect to the subband, the central frequency of the sensor's signal equals $\frac{B_{UL}}{2}$, i.e., it is not pseudorandom.
Otherwise, the packet is ``narrow'' with respect to the subband, and the sensor selects the sign $\pm$ of the frequency offset based on the frame parity. Since the frame parity changes for each frame, the probability of selecting $+$ or $-$, i.e., the upper or down half of the subband, equals 0.5. Finally, the sensors select a central frequency within the given half of the subband. The latter is parameterized by the packet Crypto Iter, which is assumed to be distributed uniformly.
Even though, according to the standard, the central frequency is defined in a discrete way (see~\eqref{eq:uplink_channel}), the frequency instability of real sensors' oscillators and the lack of synchronization before the UL transmissions can result in the real central frequency being quite far from the planned one.
For example, with $0.5$ ppm accuracy of the oscillator, the frequency deviation at \SI{868}{\MHz} can reach \SI{434}{\Hz}, while in \SI{51.2}{\kHz} channel the distance between the neighbor frequencies  approximates \SI{100}{\Hz}.
Thus, we model the central frequency as a continuous value distributed uniformly within the $[\omega_i, B_{UL} - \omega_i]$ interval.

\begin{figure}[htb]
\centering
\begin{tikzpicture}[scale=0.6]
	\draw [arrows={-triangle 45}] (0,1) -- (14,1);
	\node at (14.2,  0.8) {$f$};
	\node at (0.0,  0.6) {$0$};
	\draw (0,1) -- (0, 2.5);
	\node at (2.0,  0.6) {$\omega_i$};
	\draw [dashed] (2.0,1) -- (2.0, 2.5);
	\node at (4.0,  0.6) {$\omega_j$};
	\draw [dashed] (4.0,1) -- (4.0, 2.5);
	\node at (9.5,  0.6) {$B_{UL} - \omega_j$};
	\draw [dashed] (9.5,1) -- (9.5, 2.5);
	\node at (11.8,  0.2) {$B_{UL} - \omega_i$};
	\draw [dashed] (11.5,1) -- (11.5, 2.5);
	\node at (13.5,  0.6) {$B_{UL}$};
	\draw (13.5,1) -- (13.5, 2.5);
	\node at (5.0,  0.6) {$f_i$};
	\node at (6.8,  0.6) {$f_j$};
	\node at (5.9,  0.3) {$f_{\delta}$};
	\draw [line width=0.3mm] (4.5, 1) rectangle (5.5, 2.5);
	\draw [line width=0.3mm] (5.3, 1) rectangle (8.3, 2);
	\draw (5.0,1) -- (5.0, 2.5);
	\draw (6.8,1) -- (6.8, 2.0);
	\draw [arrows={triangle 45-triangle 45}] (5.2,0.6) -- (6.6,0.6);
\end{tikzpicture}
\vspace{-2em}
\caption{Frequency Difference}
\label{fig:delta_freq}
\end{figure}
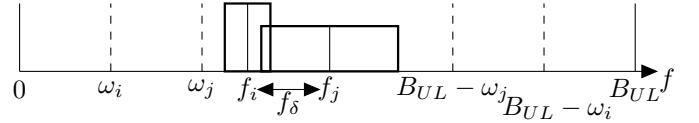

Let $D_{i, j}(f_{\delta})$  be the cumulative distribution function (CDF) of the difference between the central frequencies $f_{\delta} = |f_i - f_j|$ (see Fig.~\ref{fig:delta_freq}).
Depending on whether packets $i$ and $j$ are ``wide'' or ``narrow'', we consider four cases.

In the first case, both packets are ``wide'', i.e., $B_{UL} \le 2\omega_j$ and $B_{UL}\le 2\omega_i$, so they both have the central frequency $\frac{B_{UL}}{2}$:
\begin{equation}
\label{eq:f1}
D_{i, j}(x) = \mathbb{P}(f_{\delta} \leq x ) =
\begin{cases}
	0, & x < 0, \\
	1, & x \geq 0.
\end{cases}
\end{equation}
In the second case, $2\omega_i < B_{UL} \le 2\omega_j$, therefore $f_j = \frac{B_{UL}}{2}$ and $f_i$ is distributed uniformly over $[\omega_i, B_{UL} - \omega_i]$:
\begin{equation}
\label{eq:f2}
D_{i, j}(x) =
\begin{cases}
	0, & x < 0,\\
	\frac{2x}{B_{UL} - 2\omega_i}, & 0 \le x < \frac{B_{UL}}{2} - \omega_i,\\
	1, & x \geq \frac{B_{UL}}{2} - \omega_i.
\end{cases}
\end{equation}
In the third case, $2\omega_j < B_{UL} \le 2\omega_i$, and the result is similar to the second case with $\omega_i$ substituted with $\omega_j$.

In the fourth case, $B_{UL} > 2\omega_j$ and $B_{UL} \ge 2\omega_i$. So, $f_i$ and $f_j$ are distributed uniformly over the corresponding intervals:
\begin{equation}
\label{eq:f4}
D_{i, j}(x) =
\begin{cases}
	0, & \hspace{-2em} x < 0,\\
	\frac{2x}{B_{UL} - 2\min\left\{\omega_i, \omega_j\right\}}, & \hspace{-2em} 0 \le x < |\omega_j - \omega_i|, \\
	\frac{-x^2+2x(B_{UL} - \omega_j - \omega_i) - (\omega_j - \omega_i)^2}{(B_{UL} - 2 \omega_i)(B_{UL} - 2 \omega_j)}, &\\& \hspace{-7em} |\omega_j - \omega_i| \leq x < B_{UL} - \omega_i - \omega_j, \\
	1, & \hspace{-2em} x \geq B_{UL} - \omega_j - \omega_i.
\end{cases}
\end{equation}
Let $\mu_{i, j}(f_{\delta})$ be the probability density function (PDF) corresponding to $D_{i, j}(f_{\delta})$.

\subsection{Interference Probability}
\label{sec:interference}
In this Section, we consider a simultaneous transmission of two sensors' packets. We find the distribution of their powers at the BS which is determined by the distribution of their locations around the BS. 

We assume that the signal power from the sensor at the BS $E(r)$ depends only on the distance $r$ between them.
To find $Q_{i, j}$, we integrate over all possible locations of sensors $i$ and $j$ and over the possible differences between the packets' central frequencies where the SINR is greater than $\nu$:
\begin{multline}
\label{eq:Q_general}
Q_{i,j} = \int\limits_{r_i}\int\limits_{r_j}\rho_i(r_i)\rho_j(r_j)  \\
\times \int\limits_{f_{\delta}} \mathds{1}\left(\frac{E(r_i)}{G_{i, j}(E(r_j),f_{\delta}) + Z_i}  > \nu \right) \mu(f_{\delta}) df_{\delta}d r_j d r_i,
\end{multline}
where $\rho_i(r_i)$ is the PDF of sensor's $i$ distance from the BS, the indicator $\mathds{1}\left(x\right)$ equals $1$ if $x$ holds and $0$, otherwise, $G_{i, j}(E(r_j),f_{\delta})$ is the interference power induced by the packet $j$ at the receiver of the packet $i$, if the packet $j$ has the power $E(r_j)$ and the central frequencies of the packets $i$ and $j$ differ by $f_{\delta}$, $\nu$ is the threshold SINR value in non-dB units required for successful reception of the packet ($10^{0.7}$ or 7 dB in our scenario), and $Z_i = k T \Delta_i$ is the thermal noise in the band $\Delta_i$ of the packet $i$.
To simplify \eqref{eq:Q_general}, we assume that in the frequency domain, the power spectral density (PSD) of an NB-Fi signal is a rectangular function:

\begin{equation}
psd_i(f) = \begin{cases}
	\epsilon_i = \frac{E(r_i)}{\Delta_i}, & f_{c} - \frac{\Delta_i}{2} \leq f \leq f_{c} + \frac{\Delta_i}{2},\\
	0, & otherwise.
\end{cases}
\end{equation}
Then, the interference power equals $G_{i, j}(E(r_2),f_{\delta})=$
\begin{equation}
=
\begin{cases}
	\epsilon_j \min\left\{\Delta_i, \Delta_j\right\}, & f_{\delta} \le \frac{|\Delta_i - \Delta_j|}{2}, \\
	\epsilon_j \left( \frac{ \Delta_i + \Delta_j}{2} - f_{\delta} \right), & \frac{|\Delta_i - \Delta_j|}{2} < f_{\delta} < \frac{ \Delta_i + \Delta_j}{2}, \\
	0, & f_{\delta} \ge \frac{ \Delta_i + \Delta_j}{2}.
\end{cases}
\end{equation}
Here, the first condition corresponds to the case when one packet completely overlaps another packet, the second condition corresponds to the partial overlapping case, and the third condition corresponds to no overlap.
In all the cases, the interference power equals the width of the overlapping frequency interval multiplied by the power spectral density of the interfering packet.

Notably, $G_{i, j}(E(r_2),f_{\delta})$ is a decreasing function of $f_{\delta}$, and for a maximal value of $f_{\delta}$, the indicator in \eqref{eq:Q_general} equals $1$ (assuming that sensors are not assigned bitrates for which the SINR cannot reach \SI{7}{\dB} at the given distance from the BS). 
Let $\phi_{i,j}(r_1, r_2)$ be the minimal frequency for which the indicator equals $1$:
\begin{equation}
\phi_{i,j}(r_1, r_2) =
\begin{cases}
	0,	& \hspace{-6em}	\frac{E_i}{\epsilon_j min(\Delta_i, \Delta_j) + Z_{i}} > \nu, \\
	\frac{\Delta_i + \Delta_j}{2} - \min\left\{0, \frac{E_i(r_i) - Z_{i} \nu}{\epsilon_j \times \nu}\right\},	& \textrm{otherwise},
\end{cases}
\end{equation}
where in the first case, the packet $i$ can be received even when the packets completely intersect.

As a result, the success probability for the packet $i$ if it overlaps packet $j$ equals
\begin{equation}
\label{eq:q_simplified}
Q_{i,j} =
\begin{cases}
	\int\limits_{r_i}\int\limits_{r_j} \rho_i(r_i)\rho_j(r_j) \mathds{1}\left(\frac{E(r_i)}{\epsilon_j min(\Delta_i, \Delta_j) + Z_{i}}  > \nu \right) d r_j d r_i, &\\ \hspace{12em} B_{UL} < 2 \min\left\{\omega_i, \omega_j\right\}, \\
	1 - \int\limits_{r_i}\int\limits_{r_j} \rho_i(r_i)\rho_j(r_j) D_{i, j}\left(\phi_{i, j}(r_i, r_j)\right) d r_j d r_i, &\\ \hspace{12em} B_{UL} \geq 2 \min\left\{\omega_i, \omega_j\right\},
\end{cases}
\end{equation}
where in the first case, both packets $i$ and $j$ are ``wide'', their central frequency is $\frac{B_{UL}}{2}$, and thus they can intersect only entirely, so we integrate over such positions of sensors $i$ and $j$ that the condition inside the indicator holds; in the second case, at least one packet has a pseudo-random central frequency.

To calculate $Q_{i, j}$, we need a specific dependency of the signal power on the distance $E(r)$ and the PDFs of sensors distances.
In our study, we assume that the signal power at the receiver has a log-distance form:
\begin{equation}
\label{eq:propagation}
E(r) = E_T - A - B \log_{10} (r),
\end{equation}
where $A$ and $B$ are constants defined by the specific propagation model.

In the considered scenario, the sensors are distributed uniformly in a circle with a radius $R_1$.
The sensors use the maximal BNs allowed with their signal power which, in its turn, is defined by its distance from the BS.
Thus, the network is divided into concentric rings with  radii $R_{i}$, determined by the power thresholds such that $E(R_i) = S_{i}$ (see Fig.\ref{fig:radii}), and BN $i$ is assigned to all sensors with distance $r$ from the BS such that $R_{i+1} < r \le R_{i}$, where $R_5=0$.

\begin{figure}[htb]
\centering
\begin{tikzpicture}[scale=0.65]
	\draw [thick] (2,4) circle (4cm);
	\draw [thick] (2,4) circle (3cm);
	\draw [thick] (2,4) circle (2cm);
	\draw [thick] (2,4) circle (1cm);
	\draw [arrows={-triangle 45}] (2,4) -- (2.9,4.1);
	\draw [arrows={-triangle 45}] (2,4) -- (2.0,6.0);
	\draw [arrows={-triangle 45}] (2,4) -- (1.7,7.0);
	\draw [arrows={-triangle 45}] (2,4) -- (1.0,7.8);
	\node at (2.3,  3.8) {\scriptsize{$R_4$}};
	\node at (2.7,  5.3) {\scriptsize{$R_3$}};
	\node at (2.5,  6.5) {\scriptsize{$R_2$}};
	\node at (1.9,  7.5) {\scriptsize{$R_1$}};
	\node at (2,  3.3) {\scriptsize{BN $4$}};
	\node at (2,  2.3) {\scriptsize{BN $3$}};
	\node at (2,  1.3) {\scriptsize{BN $2$}};
	\node at (2,  0.3) {\scriptsize{BN $1$}};
\end{tikzpicture}
\caption{Bitrate distribution}
\label{fig:radii}
\end{figure}
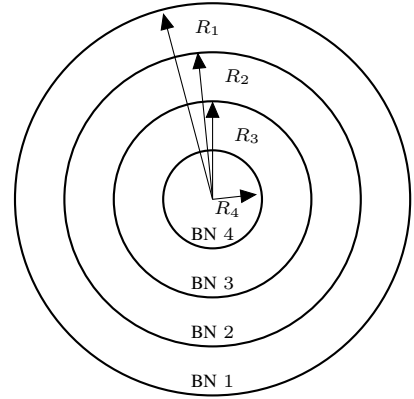

Therefore, the PDF of the sensor's distance to the BS provided that the sensor uses BN $i$ is
\begin{equation}
\rho_i (r) = \frac{2r}{R_{ i}^2 - R_{i+1}^2}.
\end{equation}
As a result, the probability of a sensor to use BN $i$ equals
\begin{equation}
p_i = \frac{R_{ i}^2 - R_{i+1}^2}{R_1^2}.
\end{equation}

With all these values, we can use \eqref{eq:q_simplified} to calculate $Q_{i, j}$ with a numerical integration algorithm, i.e., to integrate over all possible sensor distances $r_i$ and $r_j$ the PDFs of their distances and the indicators within the integrals, and then find $P^{data}_{ini, i}$.

\subsection{Probability of Successful Retry}
\label{sec:retry}
Let us find the success probability $P^{data}_{re, i}$ of a retry.
A retry occurs if a data packet is not delivered because of a collision.
In a general case, a collision may involve multiple packets.
However, under the made assumption that the traffic in the network is not heavy, we consider only the most likely case when a collision involves only two packets. 
Further derivations are done as follows.
First, we find the probability of a collision between packets transmitted at specific BNs.
Second, we consider different kinds of collisions: when both colliding packets are not received by the BS, and when only one packet is unsuccessful, and find the corresponding probabilities of successful data packet retry.
Third, we combine all these probabilities, taking into account the fact that a retry can fail due to a collision with an initial transmission by some other sensor.
Finally, we find PER for retries.

The colliding packets may be transmitted at different BNs.
Note that the probability $P^{cs}_{i, j}$ that packet $i$ collides with another packet transmitted at BN $j$ is proportional to the probability that the other packet is transmitted at BN $j$.
Similarly to \eqref{eq:pdata}: 
\begin{equation}
P^{cs}_{i, j} = \frac{\left(1 - e^{-\lambda_j (T_{i} + T_j)(1 - Q_{i,j})}\right) \prod\limits_{\substack{k = 1 \\ k \neq j}}^{4}e^{- \lambda_k (T_{i} + T_k)(1 - Q_{i,k})}}{\sum\limits_{l = 1}^{4}\left(1 - e^{-\lambda_j (T_{i} + T_l)(1 - Q_{i,l})}\right) \prod\limits_{\substack{k = 1 \\ k \neq l}}^{4} e^{- \lambda_k (T_{i} + T_k)(1 - Q_{i,k})}}.
\end{equation}

If packets $i$ and $j$ collide, the following events are possible.
\paragraph*{Event ``Packet $i$ is lost''} The packet $i$ is lost, while the packet $j$ is received successfully.
The probability of this event equals
\begin{multline}
Q^{one}_{i, j} = \int\limits_{r_i}\int\limits_{r_j} \rho_i(r_i)\rho_j(r_j) \int\limits_{f_{\delta}} \mathds{1} \left(\frac{E(r_i)}{G_{i, j}(E(r_j),f_{\delta}) + Z_i}  < \nu \right. \\
\wedge \left.\frac{E(r_j)}{G_{j, i}(E(r_i),f_{\delta}) + Z_j}  > \nu \right) \mu(f_{\delta}) df_{\delta}d r_j d r_i,
\end{multline}
which can be simplified in the same way as $Q_{i, j}$.
If both packets are ``wide'', i.e., $B_{UL} < 2 \min\left\{\omega_i, \omega_j\right\}$, then we integrate over such locations of sensors $i$ and $j$ that the condition inside the indicator $\mathds{1}()$ holds:
\begin{multline}
Q^{one}_{i, j} = \int\limits_{r_i}\int\limits_{r_j} \rho_i(r_i)\rho_j(r_j) \mathds{1}\left(\frac{E(r_i)}{\epsilon_j min(\Delta_i, \Delta_j) + Z_{i}}  < \nu \right. \\
\wedge \left.\frac{E(r_j)}{\epsilon_i min(\Delta_i, \Delta_j) + Z_{j}} > \nu\right) d r_j d r_i.
\end{multline}
Otherwise, we obtain
\begin{multline}
Q^{one}_{i, j} = \int\limits_{r_i}\int\limits_{r_j} \rho_i(r_i)\rho_j(r_j) \\
\times\max\left\{D_{i, j}(\phi_{i, j}(r_i, r_j)) - D_{i, j}(\phi_{j, i}(r_j, r_i)), 0\right\}d r_j d r_i,
\end{multline}

If only packet $i$ is lost, only the sensor $i$ makes a retry, and the probability of the successful retry equals $P^{data}_{ini, i}$, because during the retransmitted packet $i$ can only collide with other sensors' initial attempts.
The situation is different when both packets are unsuccessful.

\paragraph*{Event  ``Both packets are lost''}
The probability of this event is
\begin{equation}
Q^{both}_{i, j} = 1 - Q_{i, j} - Q^{one}_{i, j}.
\end{equation}
When packets $i$ and $j$ are lost, the sensors $i$ and $j$ make retries, and the collision probability during such retries differs significantly from that during the initial transmission attempt because a Poisson flow does not describe the retries. Let us find the probability of repetitive collision of two packets, i.e., the event when the retries intersect both in time and frequency.

Let $t = 0$ be the time when a half of the packet $i$ is transmitted (see Fig.~\ref{fig:retry}), and $x \in [-\frac{T^{Data}_{i}+T^{Data}_{j}}{2}, \frac{T^{Data}_{i}+T^{Data}_{j}}{2}]$ be the time when a half of the packet $j$ is transmitted.
As the packets $i$ and $j$ collide, the sensors make retries after a random delay.
Let $y$ and $z$ be the times when halves of packets $i$ and $j$, are retransmitted.
The time $y$ is distributed uniformly over the interval $[W_{min, i}, W_{max, i}]$, where $W_{min, i} = T_{delay, i} + T_{listen, i}$ and $W_{max, i} = W_{min, i} + T_{rnd, i}$,
$T_{delay, i}$, $T_{listen, i}$ and $T_{rnd, i}$ are defined in Table~\ref{tab:mcs_values}.
The time $z$ is distributed uniformly over the interval $[x + W_{min, j}, x + W_{max, j}]$.

A new collision happens when a packet intersects with another packet in time, the indicator of such an event being
\begin{multline}
\mathcal{I}_{i, j}\left(y,z\right) = \mathds{1}\left(y \le z \le y + \frac{T^{Data}_{i}+T^{Data}_{j}}{2}\right) + \\
+ \mathds{1}\left(z \le y \le z + \frac{T^{Data}_{i}+T^{Data}_{j}}{2} \right).
\end{multline}
With this indicator, we obtain the probability of the repetitive intersection of two packets in time:
\begin{multline}
P_{i, j}^{int} = \int\limits_{-\frac{T^{Data}_{i}+T^{Data}_{j}}{2}}^{\frac{T^{Data}_{i}+T^{Data}_{j}}{2}} \frac{1}{T^{Data}_{i}+T^{Data}_{j}} \times \\
\int\limits_{W_{min,i}}^{W_{max,i}} \int\limits_{x+W_{min,j}}^{x+W_{max,j}}\hspace{-2em} \frac{\mathcal{I}_{i, j}\left(y,z\right)}{(W_{max,i} - W_{min,i} )(W_{max,j} - W_{min,j})} dzdydx,
\end{multline}
where we integrate $\mathcal{I}_{i, j}\left(y,z\right)$ over all possible values of $x$, $y$ and $z$, taking into account their distributions.

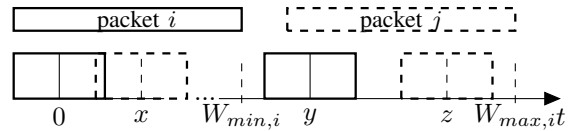
\begin{figure}[htb]
\centering
\begin{tikzpicture}[scale=0.6]
	\draw (0,1) -- (3.3,1);
	\draw [arrows={-triangle 45}] (4.5,1) -- (12,1);
	\node at (12,  0.6) {$t$};
	\node at (1.0,  0.6) {$0$};
	\node at (2.8,  0.6) {$x$};
	\node at (4.2,  1.0) {$...$};
	\node at (5.0,  0.6) {$W_{min, i}$};
	\node at (6.5,  0.6) {$y$};
	\node at (9.5,  0.6) {$z$};
	\node at (11.0,  0.6) {$W_{max, i}$};
	\draw [line width=0.3mm] (0, 1) rectangle (2.0, 2);
	\draw (1.0,1) -- (1.0, 2.0);
	\draw [dashed] [line width=0.3mm] (1.8, 1) rectangle (3.8, 2);
	\draw [dashed] (2.8,1) -- (2.8, 2.0);
	\draw [line width=0.3mm] (0.0, 2.5) rectangle (5, 3.0);
	\node at (2.7,  2.7) {\small{packet $i$}};
	\draw [dashed, line width=0.3mm] (6.0, 2.5) rectangle (11, 3.0);
	\node at (8.5,  2.7) {\small{packet $j$}};
	\draw [line width=0.3mm] (5.5, 1) rectangle (7.5, 2);
	\draw (6.5,1) -- (6.5, 2.0);
	\draw [dashed] [line width=0.3mm] (8.5, 1) rectangle (10.5, 2);
	\draw [dashed] (9.5,1) -- (9.5, 2.0);
	\draw [dashed] (11.0,  0.9) -- (11.0, 1.8);
	\draw [dashed] (5.0,  0.9) -- (5.0, 1.8);
\end{tikzpicture}
\caption{Retransmission}
\label{fig:retry}
\end{figure}

Let us consider the intersection of retries in the frequency domain.
We cannot directly use the frequency difference distribution found in Section \ref{sec:prob_freq} because, according to \eqref{eq:f_uplink}, a sensor chooses the upper or down half of the subband based on the frame parity, which does not change for retransmissions.
Thus, if two packets collide, they use the same half of the subband, and a new collision is more likely to happen.

So, we modify the approach from Section \ref{sec:prob_freq} by considering that during a retry the central frequency $f_i$ is bound to the interval $[\frac{B_{UL}}{2}, B_{UL} - \omega_i]$.
Again, if both packets are ``wide'', i.e., $B_{UL} \le min[\omega_i; \omega_j]$, then the central frequency equals $\frac{B_{UL}}{2}$ and the frequency difference CDF $	D_{i, j}^{re}(x)$ is the same as in \eqref{eq:f1}.

If $\omega_i < B_{UL} \le \omega_j$, $f_j = \frac{B_{UL}}{2}$, and $f_i$ is distributed uniformly in $[\frac{B_{UL}}{2}, B_{UL} - \omega_i]$.
In this case, the frequency difference has the CDF 	$D_{i, j}^{re}(x)$, which is the same as in \eqref{eq:f2}:

If $\omega_j < B_{UL} \le \omega_i$, the CDF is as in \eqref{eq:f2} with $\omega_i$ substituted by $\omega_j$.

If $B_{UL} > \max\left\{2\omega_j, 2\omega_i\right\}$, both $f_i$ and $f_j$ are distributed uniformly over the corresponding intervals, and $D_{i, j}^{re}(x) =$
\begin{equation}
=
\begin{cases}
	0, & \hspace{-1em}x \le 0\\
	\frac{2x\left( 2B_{UL} - 4\omega_j - x \right) }{\left( B_{UL} - 2\omega_i \right) \left(  B_{UL} - 2\omega_j \right) },    &\hspace{-1em} 0 < x \le \omega_j - \omega_i\\
	\frac{4(x( B_{UL} - \omega_i - \omega_j ) -x^2 - \frac{1}{2}(\omega_j - \omega_i)^2 ) }{\left( B_{UL} - 2\omega_i \right) \left(  B_{UL} - 2\omega_j \right)}, &\hspace{-1em} \omega_j - \omega_i
	< x \le \frac{B_{UL}}{2} - \omega_j \\
	1 - \frac{\frac{1}{2}(\omega_j - \omega_i)^2   }{\left( B_{UL} - 2\omega_i \right) \left(  B_{UL} - 2\omega_j \right)},
	&\hspace{-2em} \frac{B_{UL}}{2} - \omega_j < x \le \frac{B_{UL}}{2} - \omega_i \\
	1,                              &\hspace{-1em} x > \frac{B_{UL}}{2} - \omega_i
\end{cases}
\end{equation}

Finally, we obtain the success probability of the retransmission of packet $i$, provided that it overlaps in time with packet $j$ similarly to \eqref{eq:q_simplified}:
\begin{equation}
Q^{Re}_{i,j} =
\begin{cases}
	0, & \hspace{-11em} B_{UL} < 2 \min\left\{\omega_i, \omega_j\right\}, \\
	\\
	1 - \frac{\int_{r_i}\int_{r_j} \rho_i(r_i)\rho_j(r_j) D^{Re}_{i, j}\left(\phi_{i, j}(r_i, r_j)\right) d r_j d r_i}{1 - \int_{r_i}\int_{r_j} \rho_i(r_i)\rho_j(r_j) \mathds{1}\left(\frac{E(r_i)}{\epsilon_j min(\Delta_i, \Delta_j) + Z_{i}}  > \nu \right) d r_j d r_i}, &
	\\ \hspace{11em} B_{UL} \geq 2 \min\left\{\omega_i, \omega_j\right\}.
\end{cases}
\end{equation}
Here, the first line stands for the success probability if both packets are ``wide''. In this case, they inevitably will use the same frequency for retransmission.
The second line stands for the case when the central frequency of at least one packet is spread uniformly.
This equation differs from \eqref{eq:q_simplified} in the following way.
First, we write down a conditional success probability, provided that the signal strengths make the packet $i$ to be unsuccessful if it overlaps with the packet $j$. Thus, the denominator gives the probability that this condition holds. Second, we use  $D_{i, j}^{re}(x)$ instead of $D_{i, j}(x)$ in the numerator of the fraction.

Now we can combine the obtained results to find the success probability $P^{Data}_{i,Re}$ for a retransmission of the packet $i$:
\begin{equation}
P^{Data}_{i,Re} = \sum_{j = 1}^{4} P^{cs}_{i, j} \frac{Q_{i,j}^{one} + Q^{both}_{i}\left(1-\left(1 - Q_{i,j}^{re} \right) P_i^{int}\right)}{1-Q_{i, j}}P^{Data}_{ini, i},
\end{equation}
where we sum over the possible BNs of the colliding packet.
Inside the sum, we multiply the probability $P^{cs}_{i, j}$ that a collision happens with the packet $j$ by a sum of two probabilities: the probability of collision that results in a retry of only packet $i$, and the probability of collision that results in a retry of both packets $i$ and $j$.
In the case of only one packet being retransmitted, the probability of successful transmission is $P^{Data}_{ini, i}$.
Otherwise, we also have to multiply this probability by the probability of either packet not intersecting again in time or their frequencies and powers to be such that the retry is successful.
We also divide the probabilities by $1 - Q_{i,j}$ which is the probability of condition that the powers of the packets were such that packet $i$ was not successful (and thus the retry is required).

In the end, we find the success probability of a retry as the success probability of the data packet multiplied by the success probability of the corresponding acknowledgment delivery:
\begin{equation}
\label{eq:psre}
P^S_{i, Re} = P^{Data}_{i,Re} P^{Ack}_{i, Re},
\end{equation}
As with $P^{Ack}_{i, ini}$, we assume that $P^{Ack}_{i, Re} = 1$.
To find the average PER for retries, we take into account the fact that the rate of retries is proportional not only to $p_i$, but also to the failure probability of the initial transmission attempt (i.e., the probability of the retry to happen):
\begin{equation}
\label{eq:perre}
PER_{Re} = 1 - \sum_{i = 1}^{4} p_i \frac{1 - P^{S}_{ini, i}}{\sum_{j = 1}^{4} (1 - P^{S}_{ini, j})} P^{S}_{ini, re}.
\end{equation}

\subsection{Packet Loss Rate}
\label{sec:plr}
Given the probability that the initial transmission is successful and the retry is successful, we can find the PLR.
A packet is lost when it makes $RL$ unsuccessful transmission attempts or when a sensor makes an unsuccessful transmission attempt, and a new packet is generated while the sensor is transmitting the packet or is waiting for the $T_{listen}$ interval.
In the latter case, packets can also be lost if several packets are generated while waiting for packet transmission: in this case, only the most recently generated packet is transmitted, while the others are discarded.
However, we assume that the traffic in the network has low intensity, and such a case is improbable.

For the initial transmission attempt, the probability of a packet being dropped because of the arrival of a new packet equals
\begin{equation}
P_{ini, i}^{G} = e^{-\frac{\lambda}{N} W_{min,i}},
\end{equation}
which is the probability that the sensor does not generate a packet during the $W_{min, i}$ time starting with its packet transmission.
For retries, we also have to consider that the sensor can generate a packet during the random delay time before its transmission, which yields the following loss probability:
\begin{multline}
P_{Re, i}^{G} = \frac{1}{W_{max,i} - W_{min,i}} \int_{W_{min,i}}^{W_{max,i}} e^{-\frac{\lambda}{N} x} dx =\\
= \frac{N}{(W_{max,i} - W_{min,i})\lambda} \left(e^{-\frac{\lambda}{N} W_{min, i}} - e^{-\frac{\lambda}{N} W_{max, i}}\right).
\end{multline}

Combining these probabilities with \eqref{eq:psfirst} and \eqref{eq:psre}, we obtain
\begin{multline}
PLR_i = 1 - P^{S}_{i, ini} - \left(1 - P^{S}_{i, ini}\right) P^S_{i, Re} P_{ini, i}^{G} \times\\
\times \sum_{r = 0}^{RL - 1} \left[\left(1 - P^S_{i, Re}\right) P_{ini, i}^{G}\right]^r,
\end{multline}
where we subtract from $1$ the success probability of packet delivery, which includes the success probability of the initial transmission attempt, the probability of the failed initial transmission attempt but successful delivery at the $r^{th}$ retry, and the packet not being dropped after each transmission failure.
Finally, we average the PLR over all BNs to obtain the average PLR of a sensor:
\begin{equation}
\label{eq:plr}
PLR = \sum_{i = 1}^4 p_i PLR_i.
\end{equation}

\subsection{Average Delay}
\label{sec:delay}
We can also use the obtained probabilities to find the average delay of a delivered packet.
If a packet is delivered after a successful transmission attempt, its delay equals
\begin{equation}
D_{S, i} = T^{Data}_i + T^{Ack}_i,
\end{equation}
which is just the duration of the data packet and the corresponding acknowledgment.
Each retry increases the delay by
\begin{equation}
D_{Re, i} = W_{min, i} + \frac{T_{Rnd, i}}{2},
\end{equation}
which is the time when the sensor waits for an acknowledgment (which does not arrive) and the average backoff.
Combining these values, we obtain the average delay:
\begin{multline}
Delay_i = D_{S, i} + \left(1 - P^{S}_{i, ini}\right) P^S_{i, Re} P_{ini, i}^{G}\times\\
\times \sum_{r = 0}^{RL - 1} \left(r + 1\right) D_{Re, i} \left[\left(1 - P^S_{i, Re}\right) P_{ini, i}^{G}\right]^r.
\end{multline}
Finally, since we measure the delay only for delivered packets, we average the delay over all BNs proportionally to the delivery ratio:
\begin{equation}
\label{eq:delay}
Delay = \sum_{i = 1}^4 \frac{p_i \left(1 - PLR_i\right)}{1 - PLR} Delay_i.
\end{equation}

\subsection{Optimization Problem}
\label{sec:optimization}
Although not explicitly written, \eqref{eq:plr} and \eqref{eq:delay} define the average PLR and delay as functions of the radii $R_{ i}$.
Let $\vec{R}_{max}$ be the vector of $R_{ i}$ values.
Thus, the problem stated in Section \ref{sec:problem} can be rewritten as follows:
\begin{equation}
\label{eq:opt_plr}
\begin{split}
	\min_{\vec{R}_{max}} & \quad \mathcal{F}(\vec{R}_{max}),\\
	s.t. \quad R_{ i+1} &\leq R_{ i}, i = 1, 2, 3,\\
	\quad R_{ i} &\leq R^*_{ i}, i = 1, 2, 3, 4\\
\end{split}
\end{equation}
where $\mathcal{F}$ is either $PLR$ or $Delay$ and  $R^*_{ i}$ are given in Table ~\ref{tab:mcs_values}.
This value is found by solving the equation $\frac{E(R_{ i})}{Z_i} = \nu$.
This problem is solved numerically by searching over the possible $\vec{R}_{max} \in [0, R]^4$ values taking into account the limits of these values specified in the optimization problem~\eqref{eq:opt_plr}.

\section{Numerical Results}
\label{sec:results}
We evaluate the performance of NB-Fi networks using the developed mathematical model and the simulation.
For that, we have developed a discrete-event simulator that implements the scenario described in Section \ref{sec:problem} and does not introduce many assumptions of the mathematical model. For example, it does not assume independence of packet intersections as in Theorem~\ref{th:pdata}, neglect the influence of retries on first transmission attempts, or limit the number of colliding packets to two.
Thus, it can be used to validate the developed mathematical model.

All sensors transmit their signals with power $E_T = 14$~dBm. 
We use the Okumura-Hata model \cite{hata1980empirical} to evaluate signal propagation.
Thus the constants in \eqref{eq:propagation} are $A = 69.55 - 26.16 \log_{10} (f) + 13.82 \log_{10} (h_B) + 3.2(\log_{10} (11.75h_M))^2 - 4.97$ and $B = 44.9 - 6.55\log_{10}(h_B)$, where $h_B = \SI{30}{\m}$ and $h_M = \SI{1}{\m}$ are the heights of the BS and the sensor antennas, and $f$ is the transmission frequency (MHz).

We set the SINR threshold for successful frame reception to \SI{7}{\dB}, where \SI{2}{\dB} is the BS noise factor, and \SI{5}{\dB} is the $SNR$ required to achieve reliable transmission. 

We consider scenarios with different $R_1$ (see Section~\ref{sec:problem}). First, in the \emph{Small Circle Scenario} $R_1 = 1$~km and, thus, all sensors can potentially use any bitrate. Second, in the \emph{Average Circle Scenario} $R_1 = 5$~km, and, therefore, only sensors close to the BS can use BNs 3 and 4. Third, in the \emph{Big Circle Scenario} $R_1 = 7$~km, thus, sensors that are far from the BS can use only BN 1, while BNs 2, 3, and 4 are available only to the sensors close enough to the BS.

\subsection{Small Circle Scenario, Validation}
Consider the Small Circle Scenario.
Let the band of \SI{51.2}{\kHz} be allocated for all UL transmissions.

\begin{figure}[tb]
\centering
\includegraphics[width=0.9\linewidth]{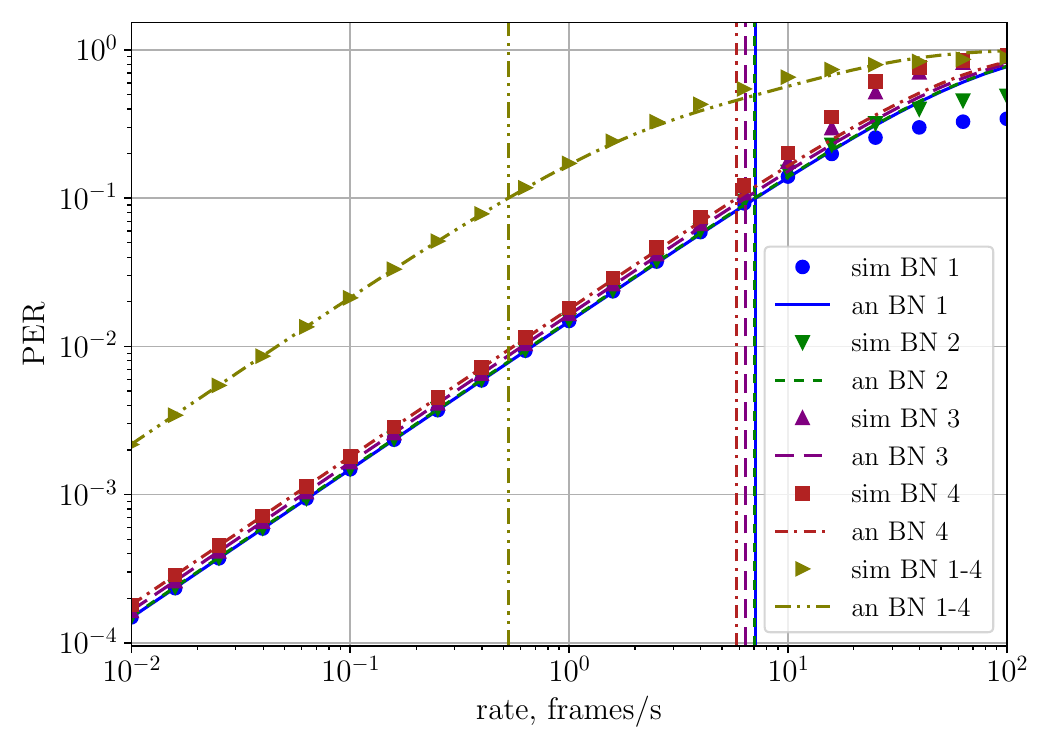}
\includegraphics[width=0.9\linewidth]{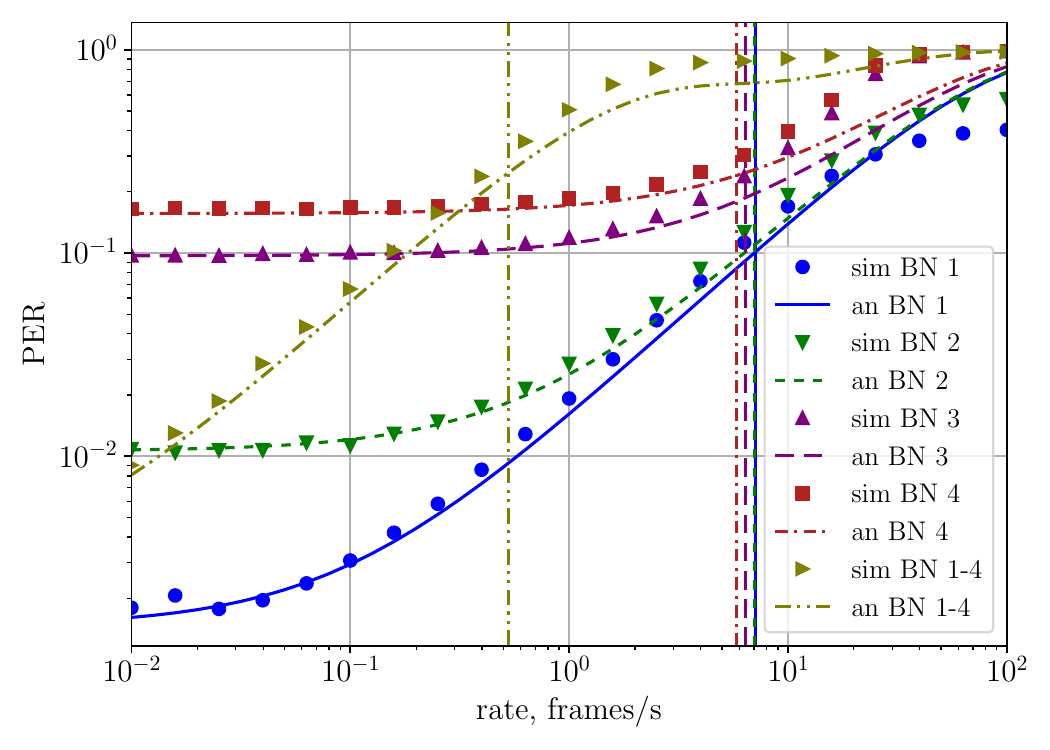}
\caption{Dependency of PER on traffic rate for initial transmission attempts (above) and retries (below), $R = \SI{1}{\km}$. Vertical lines show the accuracy bounds of the model.}
\label{fig:r_1_per}
\vspace{-1em}
\end{figure}

We start with the validation of the developed mathematical model.
Figure~\ref{fig:r_1_per} shows how PER for the initial transmission attempt and for the retries depends on the traffic rate when all sensors in the network use the same bitrate, 
or the radii of the rings make BNs distributed evenly: $p_i = \frac{1}{4}$ 
(labeled as ``BN 1--4'').
Here and further, for each value of $\lambda$, we make 100 simulation runs and each run lasts $10^7 / \lambda$ seconds.
Thus, on average, each sensor generates $10^4$ packets per run.
As we can see, the developed mathematical model is rather accurate for low traffic rates, while for high traffic rates ($\lambda > 10$ fps), the results of the mathematical model diverge from the simulation, because the assumptions of the model about the prevalence of initial transmission attempts and rare retries do not hold.
At the same time, for such a high traffic rate, PER is extremely high, which is hardly relevant to the scenario of interest.

We notice that the model is rather accurate for rates less than $\lambda^*$, which is found as the root of the equation
\begin{equation}
PER_{ini} (\lambda^*) = PER_{bound},
\end{equation}
where $PER_{bound} = 10^{-1}$ is an empiric PER bound: for PER values less than $PER_{bound}$, most packets are delivered at the initial attempt.
The accuracy bounds are shown on all plots with vertical lines for each bitrate assignment case.
Also note that the duty cycle is below 1\%, which is a typical limit for ISM bands, for all $\lambda$ such that the model is accurate.

Figure ~\ref{fig:r_1_per} shows that when all devices in the network use the same bitrate and the traffic has low intensity, the PER is approximately the same for all the bitrates, while with the mixed assignment, the PER is much higher.
Such an effect occurs because PER depends on the probability of packets intersecting in the frequency and time domains.
When devices use the same bitrate, they either transmit ``long narrow'' packets or ``short wide'' ones, which reduces intersection probability compared with the mixed assignment when ``long packets'' may intersect ``wide'', increasing PER.

Note that PER for the initial transmission attempt is much lower than for retries and also that PER during retries significantly increases with bitrate.
The latter issue could be fixed by increasing the retry time window, i.e., $T_{listen}$ and $T_{rnd}$.

\begin{figure}[tb]
\centering
\includegraphics[width=0.9\linewidth]{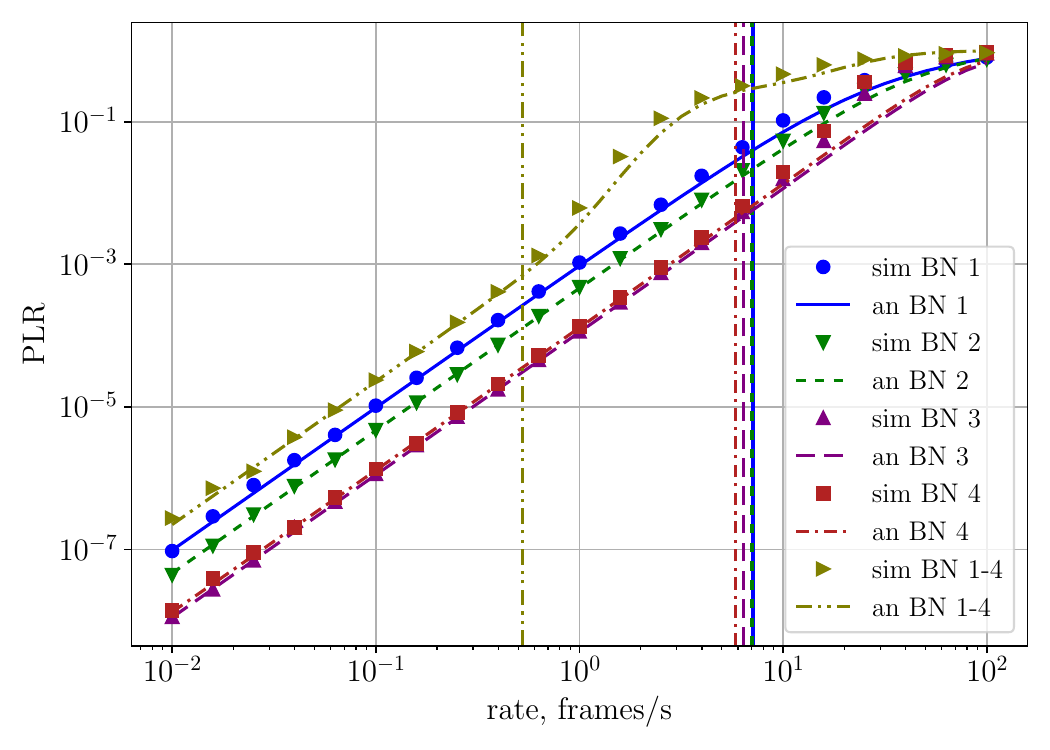}
\caption{Dependency of PLR on the traffic rate, $R = \SI{1}{\km}$. Vertical lines show the accuracy bounds of the model.}
\label{fig:r_1_plr}
\end{figure}

\subsection{Small Circle Scenario: PLR and Delay}
Figure~\ref{fig:r_1_plr} shows how PLR depends on the traffic rate.
Similar to Fig.~\ref{fig:r_1_per}, the PLR is the highest one for the mixed assignment because of the higher collision probability. In contrast to PER, PLR visibly differs for the two lowest bitrates, which have high values of the retry timings $T_{listen}$ and $T_{rnd}$.
A higher retry delay yields a higher probability of a new packet generation during the transmission attempt, which results in packet loss in case of an unsuccessful transmission attempt.

\begin{figure}[tb]
\centering
\includegraphics[width=0.9\linewidth]{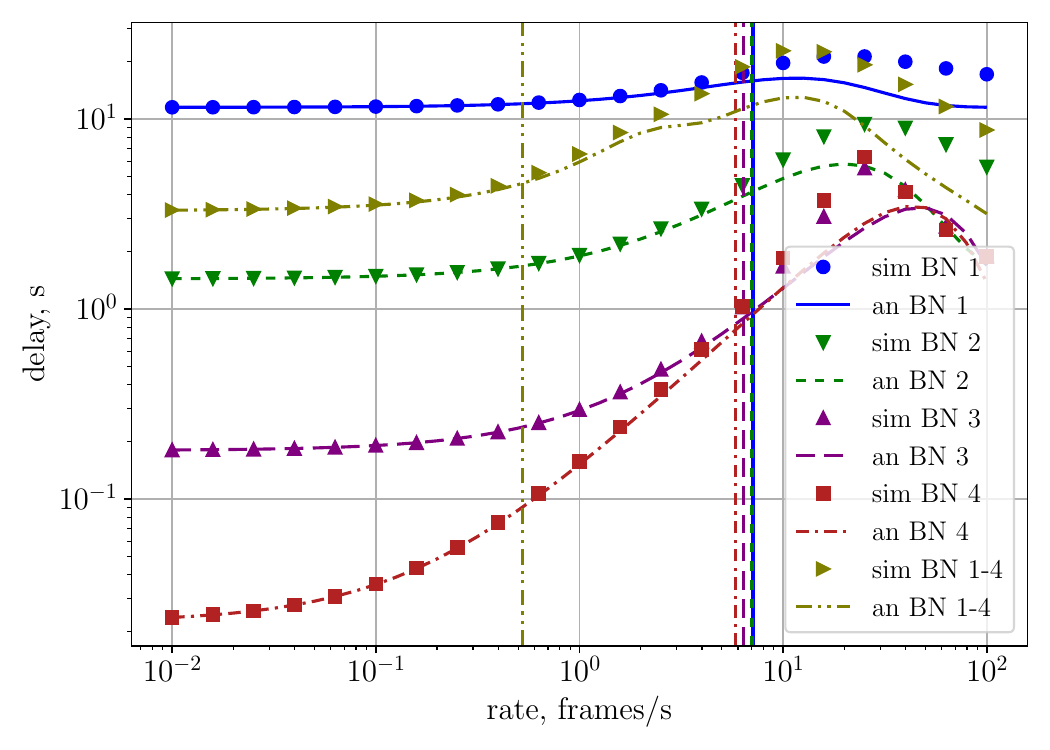}
\caption{Dependency of the average delay on the traffic rate, $R = \SI{1}{\km}$. Vertical lines show the accuracy bounds of the model.}
\label{fig:r_1_delay}
\end{figure}

Figure~\ref{fig:r_1_delay} shows how the average delay depends on the traffic rate.
For low traffic rates, the delays are impacted mainly by the duration of the data packet and acknowledgment and are inversely proportional to bitrates, while for higher traffic rates, the retry probability and retry intervals gain importance.
Thus, for the high traffic rate ($\lambda>1$ fps), the delays for BN 3 and 4 become close to each other. 

The average delays for the delivered packets have a maximal value at some high traffic rate, after which they decrease because, at a very high traffic rate, most packets are dropped. In contrast, the not dropped ones are often delivered with the only transmission attempt. 
Nevertheless, this part of the plot is out of the area of interest because it corresponds to PLR close to $1$, while for the remaining area, we see that to minimize the average delay, it is sufficient to use the highest bitrate.

\begin{figure}[tb]
\centering
\includegraphics[width=0.9\linewidth]{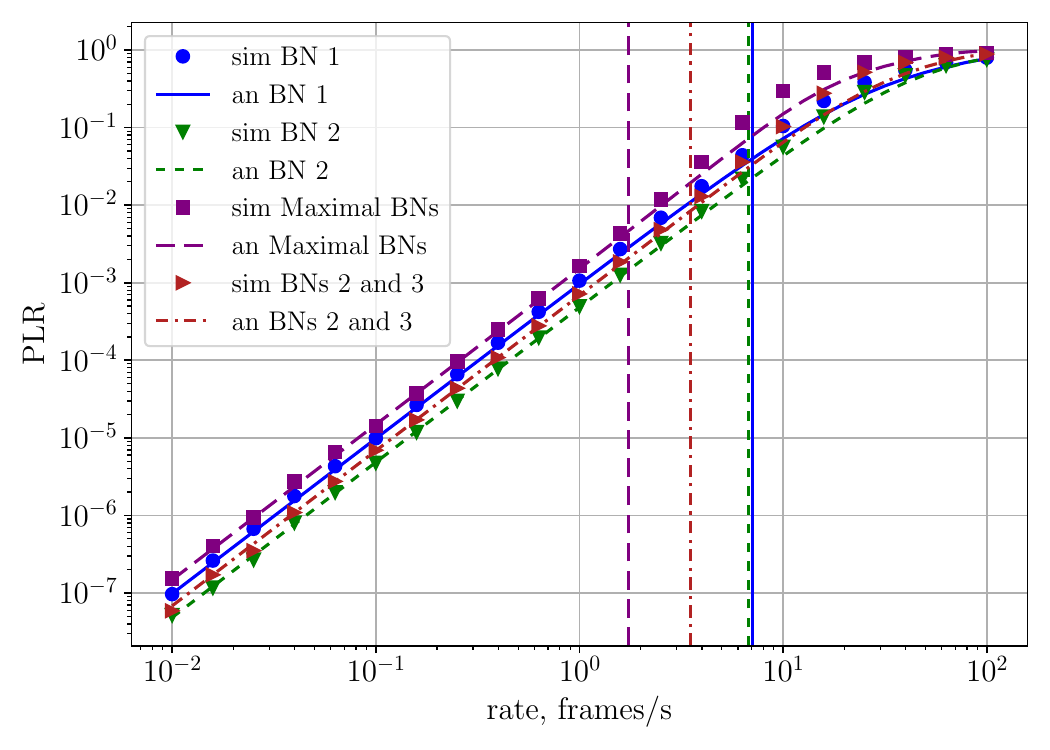}
\caption{Dependency of PLR on the traffic rate, $R = \SI{5}{\km}$. Vertical lines show the accuracy bounds of the model.}
\label{fig:r_5_plr}
\end{figure}

\subsection{Average Circle Scenario, PLR, and Delay}
Let us increase the radius of the area up to $R_1 = \SI{5}{\km}$ and consider the Average Circle Scenario. 
Figure~\ref{fig:r_5_plr} shows the dependency of PLR on the traffic rate for this scenario. 
Notably, when the sensors use the maximal possible BN according to the standardized approach described in Appendix~\ref{sec:rate_control}, the PLR is the maximal one, while using only BN 2 provides the best PLR. Even allowing some sensors close to the BS to use BN 3 (curves ``BNs 2 and 3'') only worsens the performance. 
It means that replacing the standardized algorithm with the centralized assignment of the same bitrate for all sensors improves PLR. 

\begin{figure}[tb]
\centering
\includegraphics[width=0.9\linewidth]{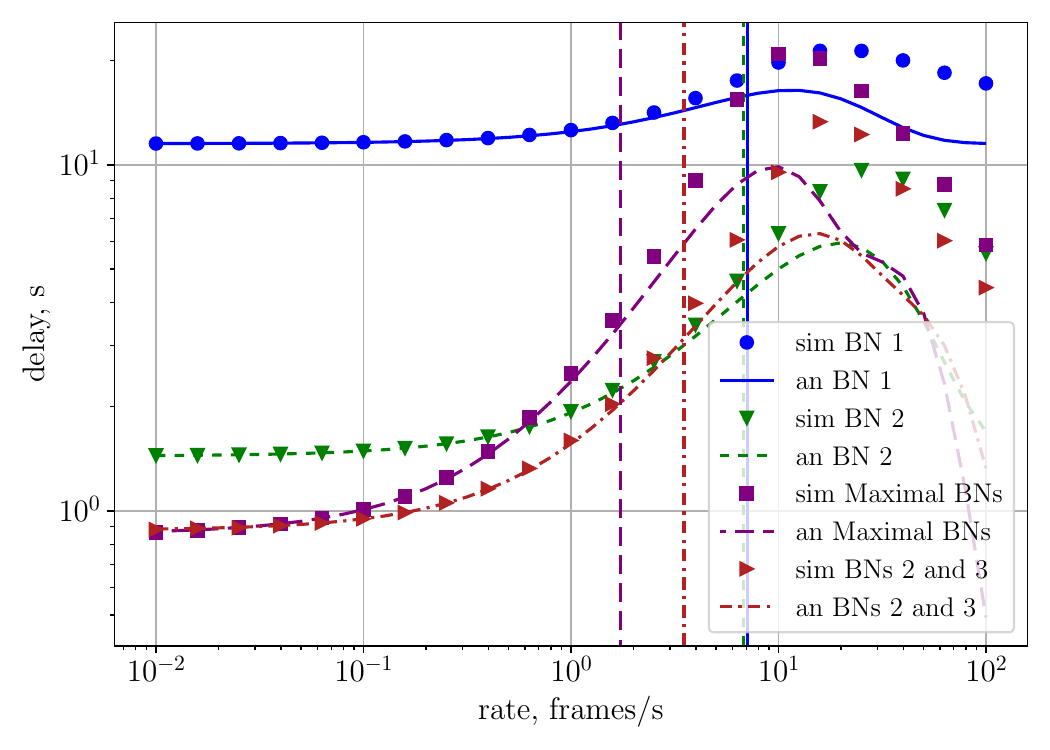}
\caption{Dependency of delay on the traffic rate, $R = \SI{5}{\km}$. Vertical lines show the accuracy bounds of the model.}
\label{fig:r_5_delay}
\end{figure}

At the same time, as Fig.~\ref{fig:r_5_delay} shows, for the average delay, the best strategy is to use a mixture of BN 2 and 3 while limiting the sensors to use only BN 2 increases the average delay because packet duration grows. Moreover, allowing the sensors to select the best bitrate may degrade delay because of more often retries. 

\subsection{Big Circle Scenario, PLR, and Delay}
Let us increase the radius of the area up to $R_1 = \SI{7}{\km}$ and consider the Big Circle Scenario.
Figure~\ref{fig:r_7_plr} shows the dependency of PLR on the traffic rate for this scenario.
As in the Average Circle Scenario, PLR is maximal when sensors use the maximal possible BN.
The best PLR is obtained when all sensors use BN 1, while allowing some sensors to use other BNs worsens the performance.
Thus, we see that to minimize the PLR, we need to assign to all sensors the same BN that is suitable for the edge sensors.

\begin{figure}[tb]
\centering
\includegraphics[width=0.9\linewidth]{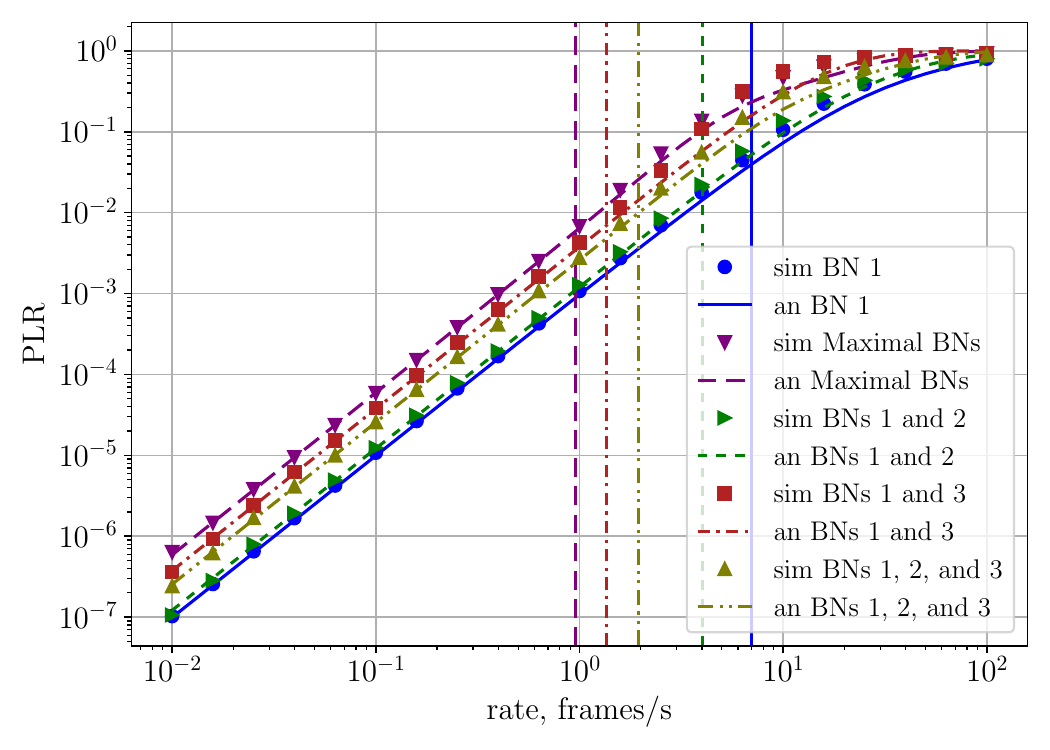}
\caption{Dependency of PLR on the traffic rate, $R = \SI{7}{\km}$. Vertical lines show the accuracy bounds of the model.}
\label{fig:r_7_plr}
\end{figure}

At the same time, this strategy is not the best one for minimizing the average delay.
As Fig.~\ref{fig:r_7_delay} shows, the best strategy for delays depends on the rate.
If the rate is below $\approx 0.2$~fps, then the best strategy is to use a mixture of BNs 1, 2, and 3.
For higher rates, the collisions become much more frequent and it becomes inefficient to use BN 3, so the best strategy is to use a mixture of BNs 1 and 2.
At the same time, making all sensors use the fastest bitrate provides higher delays due to collisions and retries.

\begin{figure}[tb]
\centering
\includegraphics[width=0.9\linewidth]{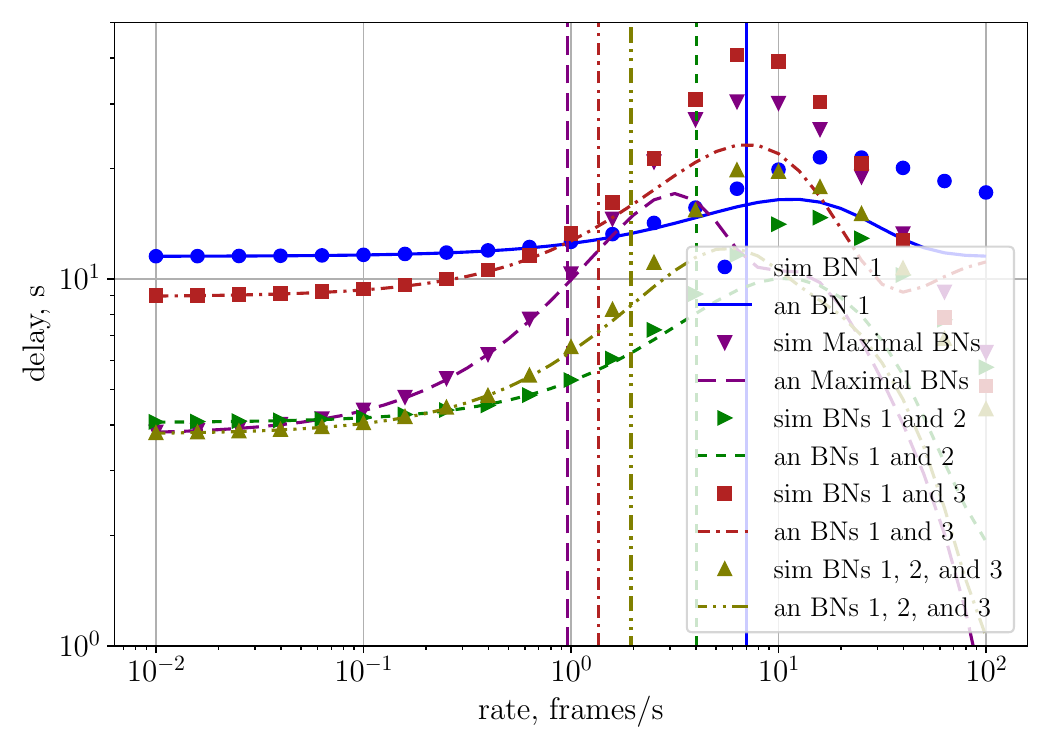}
\caption{Dependency of delay on the traffic rate, $R = \SI{7}{\km}$. Vertical lines show the accuracy bounds of the model.}
\label{fig:r_7_delay}
\end{figure}

\subsection{Guidelines for bitrate assignments}

To sum up, in order to minimize the PLR, the best strategy for all the considered scenarios is to assign the same BN to all the sensors. This BN shall be suitable for the edge sensors.
Allowing some sensors to use different BNs only increases  PLR.
The optimization of the average delay is a more sophisticated task, because the optimal delay depends on the rate.
For low rates, an efficient strategy is to assign the highest bitrates that the sensors can use with their channel conditions, however as the rate increases, it is better not to use fast bitrates because of higher collision probability.
The specific bitrate allocation that minimizes delay for the given rate can be obtained by solving the optimization task described in Section \ref{sec:optimization}.

\section{Conclusion}
\label{sec:conclusion}
In this paper, we have examined the limits of the new NB-Fi protocol, paying special attention to its features related to channel access.
We have analyzed the main features of NB-Fi and developed a mathematical model of data transmission in NB-Fi networks, which can be used to find PER, PLR, and the average delay in NB-Fi networks as functions of traffic rate and the bitrate allocation to the sensors.
Numerical results have shown that the developed model is rather accurate in the scenarios of interest. 
With our model, we have shown that depending on the scenario, the strategy to use the highest possible bitrate corresponding to the standard bitrate selection algorithm is not optimal in terms of PLR or delay, while with the developed model, we can find the optimal bitrate allocation which minimizes PLR or delay.

{\appendices
\section{Frame Format}
\label{sec:frame_format}
Figure \ref{fig:frame_ul} shows the UL frame structure.
It starts with a four-byte long preamble that serves for synchronization.
The preamble is fixed and equals 0x97157A6F.
\begin{figure}[!tb]
	\centering
	\begin{tikzpicture}[scale=0.73]
		\draw [line width=0.3mm]   (0, 0) rectangle   (2.0, 1.6);
		\draw [line width=0.3mm] (2.0, 0) rectangle  (12.0, 1.6);
		\draw [line width=0.3mm] (2.2, 0) rectangle   (4.0, 1.1);
		\draw [line width=0.3mm] (4.0, 0) rectangle   (5.8, 1.1);
		\draw [line width=0.3mm] (5.8, 0) rectangle   (7.8, 1.1);
		\draw [line width=0.3mm] (7.8, 0) rectangle  (9.8, 1.1);
		\draw [line width=0.3mm] (9.8, 0) rectangle (11.8, 1.1);
		\node at (1.0,  1.1) {\scriptsize{Preamble}};
		\node at (1.0,  0.5) {\scriptsize{4 bytes}};
		\node at (3.1,  0.7) {\scriptsize{Modem\_ID}};
		\node at (3.1,  0.3) {\scriptsize{4 bytes}};
		\node at (4.9,  0.7) {\scriptsize{Cripto Iter}};
		\node at (4.9,  0.3) {\scriptsize{1 byte}};
		\node at (6.8,  0.7) {\scriptsize{Payload}};
		\node at (6.8,  0.3) {\scriptsize{9 bytes}};
		\node at (8.8,  0.7) {\scriptsize{MIC0\_7}};
		\node at (8.8,  0.3) {\scriptsize{3 bytes}};
		\node at (10.8, 0.7) {\scriptsize{Packet CRC}};
		\node at (10.8, 0.3) {\scriptsize{3 bytes}};
		\node at (6.5,  1.3) {\scriptsize{Encoded with Error Correction Code, 32 bytes}};
	\end{tikzpicture}
	\caption{Uplink frame structure}
	\label{fig:frame_ul}
\end{figure}
\begin{figure}[!tb]
	\centering
	\begin{tikzpicture}[scale=0.73]
		\draw [dashed] (2.0, 0) rectangle  (9.4, 1.6);
		\draw [line width=0.3mm] (0.0, 0) rectangle   (2.0, 1.1);
		\draw [line width=0.3mm] (2.0, 0) rectangle   (3.8, 1.1);
		\draw [line width=0.3mm] (3.8, 0) rectangle   (5.6, 1.1);
		\draw [line width=0.3mm] (5.6, 0) rectangle   (7.4, 1.1);
		\draw [line width=0.3mm] (7.4, 0) rectangle  (9.4, 1.1);
		\draw [line width=0.3mm] (9.4, 0) rectangle (12.0, 1.1);
		\node at (1.0,  0.7) {\scriptsize{Preamble}};
		\node at (1.0,  0.3) {\scriptsize{4 bytes}};
		\node at (2.9,  0.7) {\scriptsize{Crypto Iter}};
		\node at (2.9,  0.3) {\scriptsize{1 byte}};
		\node at (4.7,  0.7) {\scriptsize{Payload}};
		\node at (4.7,  0.3) {\scriptsize{9 bytes}};
		\node at (6.5,  0.7) {\scriptsize{MIC0\_7}};
		\node at (6.5,  0.3) {\scriptsize{3 bytes}};
		\node at (8.4,  0.7) {\scriptsize{Packet CRC}};
		\node at (8.4,  0.3) {\scriptsize{3 bytes}};
		\node at (10.7, 0.7) {\scriptsize{ECC Check Bits}};
		\node at (10.7, 0.3) {\scriptsize{16 bytes}};
		\node at (5.8,  1.3) {\scriptsize{Error Correction Code (ECC) Input}};
	\end{tikzpicture}
	\caption{Downlink frame structure}
	\label{fig:frame_dl}
\end{figure}

The preamble is followed by data encoded with an error correction code, which by default is a convolutional code with a code rate $\frac{5}{8}$.
Devices can optionally use a polar core \cite{arikan2009channel}.
When a BS receives a frame, it tries to decode both codes and accepts the one that yields a valid CRC.

The encoder input consists of the following fields.

\paragraph{Modem\_ID} a 32-bit long sensor identifier.

\paragraph{Crypto Iter} eight least significant bits of the cryptographic iterator, which is a 32-bit counter stored at the transmitter and incremented by one for every transmitted frame.
NB-Fi devices use the Magma~\cite{rfc8891} symmetric key block cipher to produce security keys, encrypt the payload, and compute the message authentication code (MAC).
The basic procedure of the Magma cipher is denoted as $f_M(K, iv, p)$, where $K$ is the key, $iv$ is the initial vector, and $p$ is the plain text.
Every NB-Fi sensor has a root key $K_{root}$ installed during the production and registered at the server.
During initialization, the sensor uses the root key to generate a master key $K_{master}$ as $f_M(K_{root}, C_1, C_2)$, where $C_1$ and $C_2$ are standardized constants.
Similarly, it uses $K_{master}$ (and a different set of constants) to generate a key for data encryption $K_{data}$ and a key for MAC calculation $K_{mac}$.
Whenever Crypto Iter reaches zero (i.e., every 256 frames), the sensor generates a new master key based on the current $K_{master}$ and then uses it to generate new $K_{data}$ and $K_{mac}$.
Devices also use the cryptoiterator as the initial vector for payload encryption.

\paragraph{Payload} a 9-byte field with an encrypted payload. As the payload length is constant, all UL frames are always 36-byte long. 

\paragraph{$MIC0\_7$} three least significant bytes of the MAC, which is also calculated with the Magma algorithm in the Encrypt-then-MAC mode with the key $K_{mac}$. 
$MIC0\_7$ is used to make sure that the payload has not been changed.
Also, as described further, $MIC0\_7$ is used to calculate the central frequency of the transmission.

\paragraph{Packet CRC} three least significant bytes of the CRC32 checksum for Modem\_ID, Crypto Iter, and Payload.

Figure~\ref{fig:frame_dl} shows the DL frame structure, which is similar to the UL one with the following exceptions.

As opposed to the UL frame, the DL one does not contain Modem\_ID, and the preamble is not fixed but is calculated as a function of Modem\_ID.
Different preambles allow sensors not to react to the frames not intended for them.

Crypto Iter, Payload, Data Authentication Code, and Packet CRC form an input to the Zigzag error correction code \cite{ping2001zigzag} with the code rate of $\frac{1}{2}$.
The Zigzag code is efficient at short codeword lengths, and it has a simple decoding algorithm, which is important for the sensors.
The control bits produced by the Zigzag encoder are appended to the DL frame.

Data encryption in the DL is similar to that in the UL but uses other values for $K_{master}$, $K_{mac}$, and $K_{data}$.

\section{Transport Layer Packet Format}\label{sec:transport_format}
The payload of both UL and DL frames is a transport layer packet (see Fig. \ref{fig:header}), which consists of a header (1 byte) and data (8 bytes).
The header contains flags that specify whether the frame is system (flag SYS), e.g., the acknowledgment frame (flag ACK), whether the recipient shall send an ACK for this and, maybe, previous frames, right after the frame (flag ACK) or later (flag MULTI). Also, the header contains a 5-bit iterator (ITER) used to enumerate frames between acknowledgments.

\begin{figure}[!tb]
	\centering
	\begin{tikzpicture}[scale=0.73]
		\draw [line width=0.3mm] (0.0, 0) rectangle   (2.0, 1.1);
		\draw [line width=0.3mm] (2.0, 0) rectangle   (4.0, 1.1);
		\draw [line width=0.3mm] (4.0, 0) rectangle   (6.0, 1.1);
		\draw [line width=0.3mm] (6.0, 0) rectangle   (8.0, 1.1);
		\draw [line width=0.3mm] (8.0, 0) rectangle  (12.0, 1.1);
		\node at (1.0,  0.7) {\scriptsize{SYS}};
		\node at (1.0,  0.3) {\scriptsize{1 bit}};
		\node at (3.0,  0.7) {\scriptsize{ACK}};
		\node at (3.0,  0.3) {\scriptsize{1 bit}};
		\node at (5.0,  0.7) {\scriptsize{MULTI}};
		\node at (5.0,  0.3) {\scriptsize{1 bit}};
		\node at (7.0,  0.7) {\scriptsize{ITER}};
		\node at (7.0,  0.3) {\scriptsize{5 bits}};
		\node at (10.0,  0.7) {\scriptsize{Data}};
		\node at (10.0,  0.3) {\scriptsize{8 bytes}};
	\end{tikzpicture}
	\caption{Transport layer packet structure}
	\label{fig:header}
\end{figure}

The Data field contains application-layer payload or service information depending on the SYS flag.

\section{Optional Bitrate and Power Selection Algorithm}
\label{sec:rate_control}

Sensors operating in the DRX or CRX mode (see Section \ref{sec:modes}) may continuously evaluate the radio signal quality by estimating DL SNR while receiving frames from the BS. Also, they receive the feedback on UL SNR sent in the $ACK\_P$ frames by the BS. Consequently, the sensors may optionally adapt the bitrate and the transmission power for UL transmissions and request the most appropriate bitrate and transmission power for DL transmissions from the BS to save energy.

According to the standard, the sensors switch to a higher bitrate if the SNR reaches $SNR_{TX/RX} + SNR_{UP}$ for the UL or DL, where $SNR_{TX/RX}$ depends on the bitrate as shown in Table \ref{tab:mcs_values} and  $SNR_{UP} = \SI{15}{\dB}$ gives the minimal gap for which it is reasonable to increment the bitrate. When the sensor increments the bitrate, the bandwidth and the thermal noise increase eight times. Thus, the SNR drops by \SI{9}{\dB}, leaving \SI{6}{\dB} for fading compensation.

For the UL, the bitrate is changed without notifying the server, while for the DL, the sensor sends a service frame containing options of the new bitrate and requests an acknowledgment to this frame.

If the sensor uses the highest bitrate and SNR is high enough ($SNR \geq SNR_{UP} + SNR_{TX/RX}$ for the UL or DL), then the sensor gradually (by \SI{3}{\dB}) decreases its transmission power for UL, or notifies the server to decrease the transmission power on the BS for DL.

If the SNR is below the threshold $SNR_{DOWN} = 10$ dB, then the sensor initially increases the transmission power on the transmitter, and after the highest transmission power is reached, the sensor lowers the bitrate.

}

\bibliographystyle{IEEEtran}
\bibliography{biblio}

\end{document}